\documentclass[11pt]{article}

\usepackage{microtype}

\usepackage{fullpage}
\usepackage{amsmath,amsthm}
\usepackage{amssymb}
\usepackage{color}
\usepackage{graphicx}
\usepackage[normalem]{ulem}
\usepackage{algorithm}
\usepackage{algpseudocode}

\newtheorem{theorem}{Theorem}[section]
\newtheorem{lemma}[theorem]{Lemma}

\newtheorem{definition}[theorem]{Definition}

\newcommand{\etal}{\emph{et al.}}
\newcommand{\ball}{\mathsf{ball}}
\newcommand{\expansion}{\mathsf{expansion}}
\newcommand{\contraction}{\mathsf{contraction}}
\newcommand{\distortion}{\mathsf{distortion}}

\newcommand{\diam}{\mathsf{diam}}
\newcommand{\degmax}{\mathsf{deg}_{\mathsf{max}}}
\newcommand{\OPT}{\mathsf{OPT}}
\newcommand{\ALG}{\mathsf{ALG}}
\newcommand{\argmin}{\mathsf{argmin}}

\newcommand{\cluster}{\mathsf{CLUSTER}}
\newcommand{\path}{\mathsf{PATH}}
\newcommand{\sub}{\mathsf{SUB}}

\newcommand{\poly}{\text{poly}}

\usepackage{todonotes}
\newcommand{\cO}{\mathcal{O}}

\newcommand\tab[1][1cm]{\hspace*{#1}}

\title{Algorithms for low-distortion embeddings into arbitrary 1-dimensional spaces}

\author{
Timothy Carpenter\thanks{Dept.~of Computer Science \& Engineering, The Ohio State University.}
\and
Fedor V. Fomin\thanks{Department of Informatics, University of Bergen, Norway.}
\and
Daniel Lokshtanov\thanks{Department of Informatics, University of Bergen, Norway.}
\and
Saket Saurabh\thanks{Institute of Mathematical Sciences, India.}
\and
Anastasios Sidiropoulos\thanks{Dept.~of Mathematics and Dept.~of Computer Science \& Engineering, The Ohio State University.}
}

\begin{document}

\maketitle

\begin{abstract}
We study the problem of finding a minimum-distortion embedding of the shortest path metric of an unweighted graph into a ``simpler'' metric $X$.  Computing such an embedding (exactly or approximately) is a non-trivial task even when $X$ is the metric induced by a path, or, equivalently, into the real line. 
In this paper we give approximation and fixed-parameter tractable (FPT) algorithms for minimum-distortion embeddings into the metric of a subdivision of some fixed graph $H$, or, equivalently, into any fixed 1-dimensional simplicial complex. More precisely, we study the following problem: For  given graphs $G$, $H$  and integer $c$, is it possible to embed  $G$ with distortion $c$ into a graph homeomorphic to $H$?  
Then embedding into the line is the special case $H=K_2$, and embedding into the cycle is the case $H=K_3$, where $K_k$ denotes the complete graph on $k$ vertices.
For this problem we give 
\begin{itemize}
\item an approximation algorithm, which in time  $f(H)\cdot \poly (n)$, for some function $f$,  either correctly decides that there is no embedding of $G$ with distortion $c$ into any graph homeomorphic to $H$, or finds an embedding with distortion $\poly(c)$;
 \item an exact algorithm, which in time $f'(H, c)\cdot \poly (n)$, for some function $f'$,  either correctly decides that there is no embedding of $G$ with distortion $c$ into any graph homeomorphic to $H$, or finds an embedding with distortion $c$.
\end{itemize}

Prior to our work,  $\poly(\OPT)$-approximation or FPT algorithms were known only for embedding into paths and trees of bounded degrees.
\end{abstract}


%
%

\section{Introduction}


Embeddings of various metric spaces are a fundamental primitive in the design of algorithms \cite{Indyk01,indyk2004low,linial1995geometry,Linial02,arora2008euclidean,arora2009expander}.
A low-distortion embedding into a low-dimensional space can be used as a sparse representation of a metrical data set (see e.g.~\cite{indyk2006stable}).
Embeddings into 1- and 2-dimensional spaces also provide a natural abstraction of vizualization tasks (see e.g.~\cite{onak2013fat}).
Moreover embeddings into topologically restricted spaces can be used to discover interesting structures in a data set; for example, embedding into trees is a natural mathematical abstraction of phylogenetic reconstruction (see e.g.~\cite{farach1995robust}).
More generally, embedding into ``algorithmically easy'' spaces provides a general reduction for solving geometric optimization problems (see e.g.~\cite{bartal1996probabilistic,farach1999approximate}).

A natural algorithmic problem that has received a lot of attention in the past decade concerns the exact or approximate computation of embeddings of minimum distortion of a given metric space into some host space (or, more generally, into some space chosen from a specified family).
Despite significant efforts, most known algorithms for this important class of problems work only for the case of the real line and trees.

In this work we present exact and approximate algorithms for computing minimum distortion embeddings into arbitrary 1-dimensional topological spaces of bounded complexity.
More precisely, we obtain algorithms for embedding the shortest-path metric of a given unweighted graph into a subdivision of an arbitrary graph $H$.
The case where $H$ is just one edge is precisely the problem of embedding into the real line.
We remark that prior to our work, even the case where $H$ is a triangle, which corresponds to the problem of embedding into a cycle, was open.

We remark that the problem of embedding shortest path metrics of finite graphs into any fixed finite 1-dimensional simplicial complex ${\cal C}$ is equivalent to the problem of embedding into arbitrary subdivisions of some fixed finite graph $H$, where $H$ is the abstract 1-dimensional simplicial complex corresponding to ${\cal C}$\footnote{Here, a $d$-dimensional simplicial complex, for some integer $d\geq 1$, is the space obtained by taking a set of simplices of dimension at most $d$, and identifying pairs of faces of the same dimension.
An abstract $d$-dimensional simplicial complex ${\cal A}$ is a family of nonempty subsets of cardinality at most $d+1$ of some ground set $X$, such that for all $Y'\subset Y\in {\cal A}$, we have $Y'\in {\cal A}$;
in particular, any $1$-dimensional simplicial complex corresponds to the set of edges and vertices of some graph.}.
Since we are interested in algorithms, for the remainder of the paper we state all of our results as embeddings into subdivisions of graphs.



\subsection{Our contribution}
We now formally state our results and briefly highlight the key new techniques that we introduce.
The input space consists of some unweighted graph $G$.
The target space is some unknown subdivision $H'$ of some fixed $H$; we allow the edges in $H'$ to have arbitrary non-negative edge lengths.

We first consider the problem of approximating a minimum-distortion embedding into arbitrary $H$-subdivisions.
We obtain a polynomial-time approximation algorithm, summarized in the following.
The proof is given in Section~\ref{sec:H_approx}.

\newtheorem*{thm:mainApprox}{Theorem \ref{thm:mainApprox}}
\begin{thm:mainApprox}
There exists a $8^hn^{\cO(1)}$ time algorithm that takes as input an $n$-vertex graph $G$, a graph $H$ on $h$ vertices, and an integer $c$, and either correctly concludes that there is no $c$-embedding of $G$ into a subdivision of $H$, or produces a $c_{\ALG}$-embedding of $G$ into a subdivision of $H$, with $c_{\ALG} \leq  64 \cdot 10^6 \cdot c^{24} (h+1)^9$.
\end{thm:mainApprox}

In addition, we also obtain a FPT algorithm, parameterized by the optimal distortion and $H$.
The proof is given in Section~\ref{sec:H_fpt}.

\newtheorem*{theorem:fpt_H}{Theorem \ref{theorem:fpt_H}}
\begin{theorem:fpt_H} 
Given an integer $c>0$ and graphs $G$ and $H$, it is possible in time $f(H, c) \cdot n^{\cO(1)}$ to either find a non-contracting $c$-embedding of $G$ into a subdivision of $H$, or correctly determine that no such embedding exists.
\end{theorem:fpt_H}

\subsection{Related work}

\paragraph{Embedding into 1-dimensional spaces.}
Most of  the previous work on approximation and FPT algorithms for  low-distortion embedding (with one notable recent exception~\cite{NayyeriR16}) concerns  embeddings of a more general metric space $M$  into the real line and trees. However, even in the  case of embedding into the line,   all polynomial time approximation algorithms    make assumptions on the   metric $M$ such as having bounded spread (which is the ratio between the maximum and the minimum point distances in $M$) \cite{BadoiuCIS05,nayyeri2015reality} or being the shortest-path metric of an unweighted graph   \cite{BadoiuDGRRRS05}. 
This happens for  to a good reason: 
as it was shown by B{\u{a}}doiu {\it et al.}  \cite{BadoiuCIS05}, 
 computing the minimum line distortion is hard to approximate up to a
factor polynomial in $n$, even when $M$ is the weighted tree metrics  with spread  $n^{\cO(1)}$.  

Most relevant to our  approximation  algorithm is the work of  
B{\u{a}}doiu {\it et al.} 
\cite{BadoiuDGRRRS05}, who  
gave an 
 algorithm that for a given $n$-vertex (unweighted)  graph $G$ and $c>0$ in time  $\cO(cn^3)$ either concludes correctly that no $c$-distortion of $G$ into line exists, or computes an $\cO(c)$-embedding of $G$ into the line. 
 Similar results can be obtained for embedding into trees~\cite{BadoiuDGRRRS05,BadoiuIS07}. 
  Our approximation algorithm can be seen as an extension of these results to much more general metrics.

  
 Parameterized complexity of    low-distortion embeddings was considered by 
Fellows {\it et al.}~\cite{Fellows:2013:DFP:2539126.2489789}, who 
  gave a fixed parameter tractable (FPT) algorithm for finding an embedding of an {unweighted} graph metric   into the line with distortion at most $c$, or concludes that no such embedding exists, which works in time 
  $\cO(nc^4(2c+1)^{2c})$. As it was shown by Lokshtanov  {\it et al.}~\cite{LokshtanovMS11-superexp} that, unless ETH fails, this bound is asymptotically tight. 
  For   
{weighted} graph metrics Fellows {\it et al.}  obtained an  algorithm with running time $\cO(n(cW)^4(2c+1)^{2cW})$, where $W$ is the largest edge weight of the input graph. In addition, they   rule out, unless P=NP, any possibility of an algorithm with running time $\cO((nW)^{h(c)})$, where
$h$ is a function of $c$ alone. The problem of low-distortion embedding into a tree 
is FPT parameterized by the maximum vertex degree in the tree and the distortion $c$  \cite{Fellows:2013:DFP:2539126.2489789}.

Due to the intractability of  low-distortion embedding problems  from approximation and parameterized complexity perspective,   
Nayyeri and Raichel \cite{nayyeri2015reality} initiated the study of approximation algorithms with  running time, in the worse case, not necessarily polynomial and not even FPT.  In a very recent work
Nayyeri and Raichel \cite{NayyeriR16} obtained a  
 $(1+\varepsilon)$-approximation algorithm for finding the minimum-distortion embedding of an $n$-point metric space $M$ into the shortest path metric space of a weighted graph $H$ with $m$ vertices. The running time of their algorithm is
 $(c_\OPT \Delta)^{\omega\cdot \lambda \cdot  (1/\varepsilon)^{\lambda +2} \cdot \cO((c_\OPT)^{2\lambda}) }\cdot n^{\cO(\omega)} \cdot m^{\cO(1)}$, where $\Delta$ is the spread 
 of the points of $M$,  $\omega$ is  
 the treewidth of $H$  and $\lambda$ is the doubling dimension of $H$.
Our approximation and FPT algorithms  and  the algorithm of Nayyeri and Raichel are incomparable. 
Their algorithm is for more general metrics but runs in polynomial time  only when the optimal distortion $c_\OPT$ is constant, even when $H$ is a cycle. In contrast, our approximation algorithm runs in polynomial time for any value of $c_\OPT$.
Moreover, the  algorithm of Nayyeri and Raichel is (approximation) FPT with parameter $c_\OPT$ only when the spread $\Delta$ of $M$ (which in the case of the unweighted graph metric is the diameter of the graph) and the doubling dimension of the host space are both constants; when $c_\OPT=\cO(1)$ (which is the interesting case for FPT algorithms), this implies that the doubling dimension of $M$ must also be constant, and therefore $M$ can contain only a constant number of points, this makes the problem trivially solvable in constant time.
The running time of our parameterized algorithm does not depend on the spread of the metric of $M$. 
%

%

\paragraph{Embedding into higher dimensional spaces.}
Embeddings into $d$-dimensional Euclidean space have also been investigated.
The problem of approximating the minimum distortion in this setting appears to be significantly harder, and most known results are lower bounds \cite{matouvsek2010inapproximability,edmonds2010inapproximability}.
Specifically, it has been shown by Matou{\v{s}}ek and Sidiropoulos \cite{matouvsek2010inapproximability} that it is NP-hard to approximate the minimum distortion for embedding into $\mathbb{R}^2$ to within some polynomial factor.
Moreover, for any fixed $d\geq 3$, it is NP-hard to distinguish whether the optimum distortion is at most $\alpha$ or at least $n^\beta$, for some universal constants $\alpha,\beta>0$.
The only known positive results are a $\cO(1)$-approximation algorithm for embedding subsets of the 2-sphere into $\mathbb{R}^2$ \cite{BadoiuDGRRRS05}, and approximation algorithms for embedding ultrametrics into $\mathbb{R}^d$ \cite{badoiu2006embedding,onak2013fat}.

\paragraph{Bijective embeddings.}
We note that the approximability of minimum-distortion embeddings has also been studied for the case of bijections \cite{papadimitriou2005complexity,HallP05,KenyonRS04,khot2007hardness,ChandranMOPSS08,KenyonRS09,edmonds2010inapproximability}.
In this setting, most known algorithms work for subsets of the real line and for trees.

\section{Notation and definitions}

For a graph $G$, we denote by $V(G)$ the set of vertices of $G$ and by $E(G)$ the set of edges of $G$.
For some $U\subseteq V(G)$, we denote by $G[U]$ the subgraph of $G$ induced by $U$.
Let $\degmax(G)$ denote the maximum degree of $G$.

Let $M=(X,d)$, $M'=(X',d')$ be metric spaces.
An injective map $f:X\to X'$ is called an \emph{embedding}.
The \emph{expansion} of $f$ is defined to be
$\expansion(f) = \sup_{x'\neq y'\in X} \frac{d'(f(x'),f(y'))}{d(x',y')}$
and the \emph{contraction} of $f$ is defined to be
$\contraction(f) = \sup_{x\neq y\in X} \frac{d(x,y)}{d'(f(x),f(y))}$.
We say that $f$ is \emph{non-expanding} (resp.~\emph{non-contracting}) if $\expansion(f)\geq 1$ (resp.~$\contraction(f)\geq 1$).
The distortion of $f$ is defined to be
$\distortion(f) = \expansion(f) \cdot \contraction(f)$.
We say that $f$ is a \emph{$c$-embedding} if $\distortion(f) \leq c$.

For a metric space $M=(X,d)$, for some $x\in X$, and $r\geq 0$, we write
$\ball_M(x,r)=\{y\in X : d(x,y)\leq r\}$,
and for some $Y\subseteq X$, we define
$\diam_M(Y) = \sup_{x,y\in Y} d(x,y)$.
We omit the subscript when it is clear from the context.
We also write $\diam(M) = \diam_M(X)$.
When $M$ is finite, the \emph{local density} of $M$ is defined to be 
$\delta(M) = \max_{x \in X, r>0} \frac{|\ball_M(x,r)-1|}{2r}$.
For a graph $G$, we denote by $d_G$ the shortest-path distance in $G$.
We shall often use $G$ to refer to the metric space $(V(G), d_G)$.

For graphs $H$ and $H'$, we say that $H'$ is a \emph{subdivision} of $H$ if it is possible, after replacing every edge of $H$ by some path, to obtain a graph isomorphic to $H'$.


\section{Overview of our results and techniques}

Here we present our main theorems and algorithms, with a short discussion. Formal proofs and detailed statements of the algorithms are left to later sections in the paper.

\paragraph{Approximation algorithm for embedding into an $H$-subdivision for general $H$.}
Here, we briefly highlight the main ideas of the approximation algorithm for embedding into $H$-subdivisions, for arbitrary fixed $H$.
A key concept is that of a \emph{proper} embedding: this is an embedding where every edge of the target space is ``necessary''. In other words, for every edge $e$ of $H'$ there exists some vertices $u$, $v$ in $G$, such that the shortest path between $u$ and $v$ in $H'$ traverses $e$. Embeddings that are not proper are difficult to handle. We first guess the set of edges in $H$ such that their corresponding subdivisions in $H'$ contain unnecessary edges; we ``break'' those edges of $H$ into two new edges having a leaf as one of their endpoint. There is a bounded number of guesses (depending on $H$), and we are guaranteed that for at least one guess, there exists an optimal embedding that is proper. By appropriately scaling the length of the edges in $H'$ we may assume that the embedding we are looking for has contraction exactly $1$. The importance of using proper embeddings is that a proper embedding which is ``locally'' non-contracting is also (globally) non-contracting, while this is not necessarily true for non-proper embeddings.

A second difficulty is that we do not know the number of times that an edge in $H$ is being subdivided. Guessing the exact number of times each edge is subdivided would require $n^{f(H)}$ time, which is too much. Instead we set a specific threshold $\ell$, based on $c$. The threshold $\ell$ is approximately $c^3$, and essentially $\ell$ is a threshold for how many vertices a BFS in $G$ needs to see before it is able to distinguish between a part of $G$ that is embedded on an edge, and a part of $G$ that is embedded onto in an area of $H'$ close to a vertex of degree at least $3$. In particular, parts of $G$ that are embedded close to the middle of an edge {\em can} be embedded with low distortion onto the line, while parts that are embedded close to a vertex of degree $3$ in $H$ can not - because $G$ ``grows in at least $3$ different directions'' in such parts. Since BFS is can be used as an approximation algorithm for embedding into the line, it will detect whether the considered part of $G$ is close to a degree $\geq 3$ vertex of $H$ or not.

Instead of guessing exactly how many times each edge of $H$ is subdivided, we guess for every edge whether it is subdivided at least $\ell$ times or not. The edges of $H$ that are subdivided at least $\ell$ times are called ``long'', while the edges that are subdivided less than $\ell$ times are called ``short''. We call the connected components of $H$ induced on the short edges a {\em cluster}. Having defined clusters, we now observe that a cluster with only two long edges leaving it {\em can} be embedded into the line with (relatively) low distortion, contradicting what we said in the previous paragraph! Indeed, the parts of $G$ mapped to a cluster with only two long edges leaving it are (from the perspective of a BFS), indistinguishable from the parts that are mapped in the middle of an edge! For this reason, we classify clusters into two types: the {\em boring} ones that have at most two (long) edges leaving them, and the {\em interesting} ones that are incident to at least $3$ long edges.

Any graph can be partitioned into vertices of degree at least $3$ and paths between these vertices such that every internal vertex on these paths has degree $2$. Thinking of clusters as ``large'' vertices and the long edges as edges between clusters, we can now partition the ``cluster graph'' into interesting clusters (i.e vertices of degree 3), and chains of boring clusters between the interesting clusters -- these chains correspond to paths of vertices of degree $2$.

The parts of $G$ that are embedded onto a chain of boring clusters can be embedded into the line with low distortion, and therefore, for a BFS these parts are indistinguishable from the parts of $G$ that are embedded onto a single long edge. However, the interesting clusters are distinguishable from the boring ones, and from the parts of $G$ that are mapped onto long edges, because around interesting clusters the graph really does  ``grow in at least $3$ different directions'' for a long enough time for a BFS to pick up on this.

Using the insights above, we can find a set $F$ of at most $|V(H)|$ vertices in $G$, such that every vertex in $F$ is mapped ``close'' to some interesting cluster, and such that every interesting cluster has some vertex in $F$ mapped ``close'' to it. At this point, one can essentially just guess in time $\cO(h^h)$ which vertex of $F$ is mapped close to which clusters of $H$. Then one maps each of the vertices that are ``close'' to $F$ (in $G$) to some arbitrarily chosen spot in $H$ which is close enough to the image of the corresponding vertex of $F$. Local density arguments show that there are not too many vertices in $G$ that are ``close'' to $F$, and therefore this arbitrary choice will not drive the distortion of the computed mapping up too much.

It remains to embed all of the vertices that are ``far'' from $F$ in $G$. However, by the choice of $F$ we know that all such vertices should be embedded onto long edges, or onto chains of boring clusters. Thus, each of the yet un-embedded parts of the graph can be embedded with low distortion into the line! All that remains is to compute such low distortion embeddings for each part using a BFS, and assign each part to an edge of $H$. Stitching all of these embeddings together yields the approximation algorithm.

There are multiple important details that we have completely ignored in the above exposition. The most important one is that a cluster can actually be quite large when compared to a long edge. After all, a boring cluster contains up to $E(H)$ short edges, and the longest short edge can be almost as long as the shortest long edge! This creates several technical complications in the algorithm that computes the set $F$. Resolving these technical complications ends up making it unnecessary to guess which vertex of $F$ is mapped to which vertex of $H$, instead one can compute this directly, at the cost of increasing the approximation ratio.

\paragraph{FPT algorithm for embedding into an $H$-subdivision for general $H$.}

Our FPT algorithm for embedding graphs $G$ into $H$-subdivisions (for arbitrary fixed $H$) draws inspiration from the algorithm for the line used in \cite{FellowsFLLSR09,BadoiuDGRRRS05}, while also using an approach similar to the approximation algorithm for $H$-subdivisions. The result here is an exact algorithm with running time $f(H,c_\OPT)\cdot n^{\cO(1)}$. 

A naive generalization of the algorithm for the line needs to maintain the partial solution over $f(H)$ intervals, which results in running time $n^{g(H)}$, which is too much. Supposing that there is a proper $c$-embedding of $G$ into some $H$-subdivision, we attempt to find this embedding by guessing the short and long edges of $H$.
Using this guess, we partitions $H$ into connected clusters of short and long edges (we call the clusters of short edges ``interesting'' clusters, and the clusters of long edges ``path'' clusters). We show that if a $c$-embedding exists, we can find a subset of $V(G)$, with size bounded by a function of $|H|$ and $c$, that contains all vertices embedded into the interesting clusters of $H$. From this, we make further guesses as to which specific vertices are embedded into which interesting clusters, then how they are embedded into the interesting clusters. We also make guesses as to what the embedding looks like for a short distance (for example, $\cO(c^2)$) along the long edges which are connected to the important clusters.

\begin{figure}[H]
	\centering
  \includegraphics[width=0.4\textwidth]{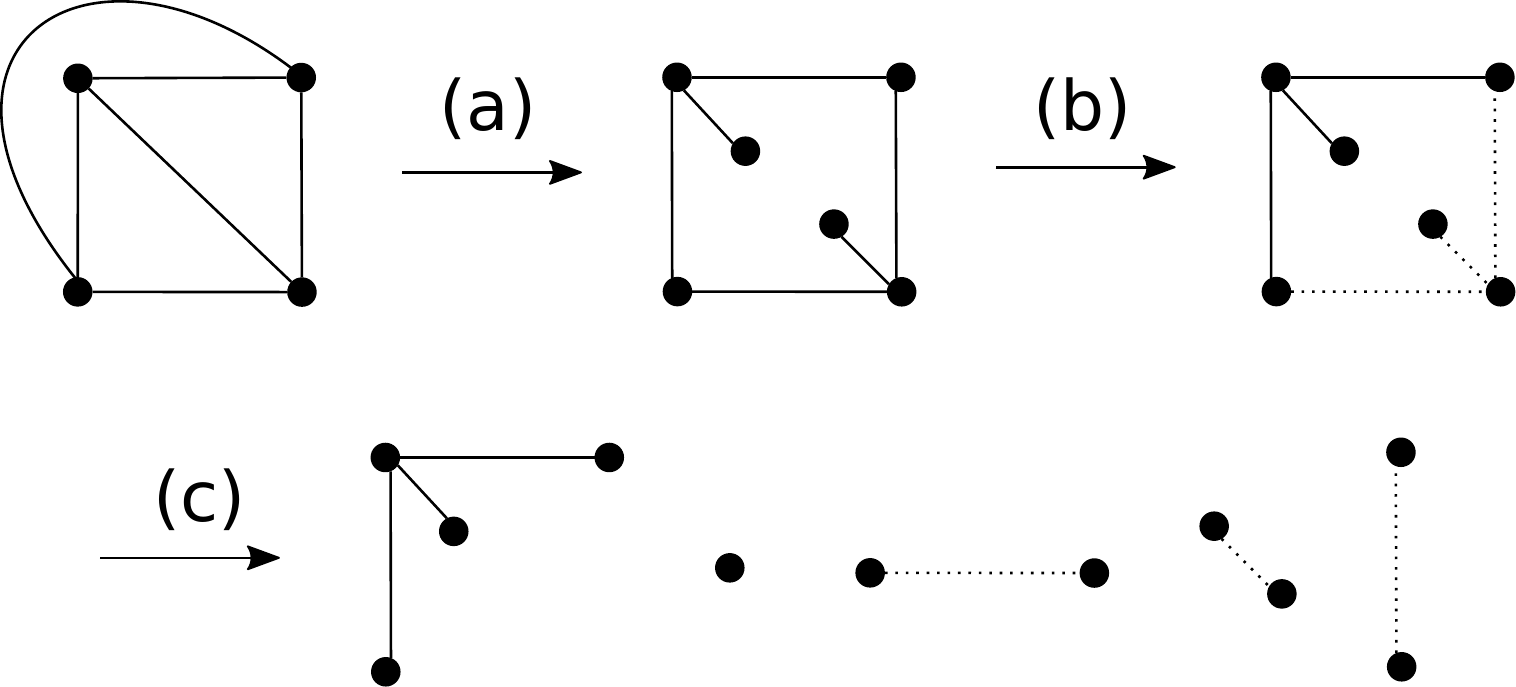}
  \caption{The FPT algorithm follows this process: (a) A quasi-subgraph of the target graph is chosen. (b) Short (solid line) and long (dotted line) edges are chosen. (c) The graph is divided into interesting (left 2) and path (right 3) components.}
  \label{fig-fpt}
\end{figure}

Since the number of guesses at each step so far can be bounded in terms of $c$ and $H$, we can iterate over all possible configurations. Once our guesses have found the correct choices for the interesting clusters and for a short distance along the paths leaving these clusters, we are able to partition the remaining vertices of $G$, and guess which path clusters these partitions are embedded into. Due to the ``path-like'' nature of the path clusters, when we pair the correct partition and path cluster, we are able to use an approach inspired by \cite{FellowsFLLSR09,BadoiuDGRRRS05} to find a $c$-embedding of the partition into the path cluster, which is compatible with the choices already made for the interesting clusters.  The formal description and analysis of this algorithm is quite lengthy, and deferred to Section \ref{sec:H_fpt}.

\section{Preliminaries on embeddings into general graphs}

Let $G$, $H$ be connected graphs, with a fixed total order $<$ on $V(G)$ and $V(H)$. A non-contracting, $c_\OPT$-embedding of $G$ to $H$ is a function $f_\OPT : V(G) \rightarrow (H_\OPT, w_\OPT)$, where $H_\OPT$ is a subdivision of $H$, $w_\OPT : E(H_\OPT) \rightarrow \mathbb{R}^{>0}$, and for all $u, v \in V(G)$,
\[
d_G(u, v) \leq d_{(H_\OPT, w_\OPT)}(f_\OPT(u), f_\OPT(v)) \leq c_\OPT \cdot d_G(u, v),
\]
where $d_{(H_\OPT, w_\OPT)}$ is the shortest path distance in $H_\OPT$ with respect to $w_\OPT$. Stated formally, for all $h_1, h_2 \in V(H_\OPT)$, if $\mathcal{P}$ is the set of all paths from $h_1$ to $h_2$ in $H_\OPT$, then
\[
d_{(H_\OPT, w_\OPT)}(h_1, h_2) = \min_{P \in \mathcal{P}}\left\{\sum_{e \in P} w_\OPT(e)\right\}.
\]

\begin{definition}
For a graph $G_1$ and subdivision $G'_1$ of $G_1$, for $e \in E(G_1)$, let $\sub_{G'_1}(e)$ be the subdivision of $e$ in $G'_1$. For convenience, for each $e \in E(H)$, we shall use $e_\OPT$ to indicate the subdivision of $e$ in $H_\OPT$.
\end{definition}

The following notion of consecutive vertices will be necessary to describe additional properties we will want our embeddings to have.

\begin{definition} 
Suppose there exists $u, v \in V(G)$ and $e \in E(H)$ such that $f_\OPT(u), f_\OPT(v) \in V(e_\OPT)$ and $f_\OPT(u) < f_\OPT(v)$.
If for all $w \in V(G) \setminus \{u, v\}$, $f_\OPT(w)$ is not in the path in $e_\OPT$ between $f_\OPT(u)$ and $f_\OPT(v)$, then we say that $u$ and $v$ are \emph{consecutive w.r.t.}~$e$, or we say that $u$ and $v$ are \emph{consecutive}.
\end{definition}

The first property we will want our embeddings to have is that they are ``pushing''. The intuition here is that we want our embedding to be such that we cannot modify the embedding by contracting the distance further between two consecutive vertices.

\begin{definition}
If for all $u, v \in V(G)$ and $e \in E(H)$ such that $u$ and $v$ are consecutive w.r.t.~$e$ we have that
\[
d_{e_\OPT}(f_\OPT(u),f_\OPT(v)) = d_G(u, v),
\]
then we say that $f_\OPT$ is \emph{pushing}.
\end{definition}

\begin{figure}[H]
	\centering
  \includegraphics{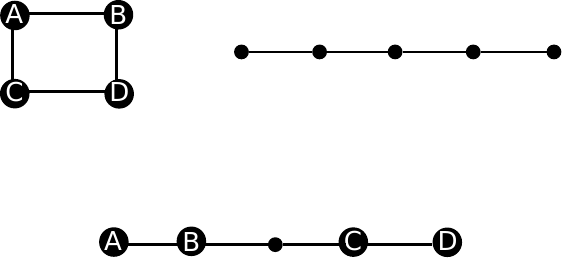}
  \caption{A pushing embedding of the cycle on 4 vertices into a path of length 5.}
  \label{fig-pushing}
\end{figure}

The next property we want for our embeddigns is that they are ``proper'', meaning that all edges of the target graph are, in a loose sense, covered by an edge of the source graph.

\begin{definition}
For any $z \in V(H_\OPT)$, if there exists $\{u, v\} \in E(G)$ such that
\[
d_{(H_\OPT, w_\OPT)}(f_\OPT(u), f_\OPT(v)) = d_{(H_\OPT, w_\OPT)}(z, f_\OPT(u)) + d_{(H_\OPT, w_\OPT)}(z, f_\OPT(v))
\]
then we say that $z$ is \emph{proper w.r.t.~$f_\OPT$}. If for all $x \in V(G)$, $x$ is proper w.r.t.~$f_\OPT$, then we say that $f_\OPT$ is \emph{proper}.
\end{definition}

\begin{figure}[H]
	\centering
  \includegraphics{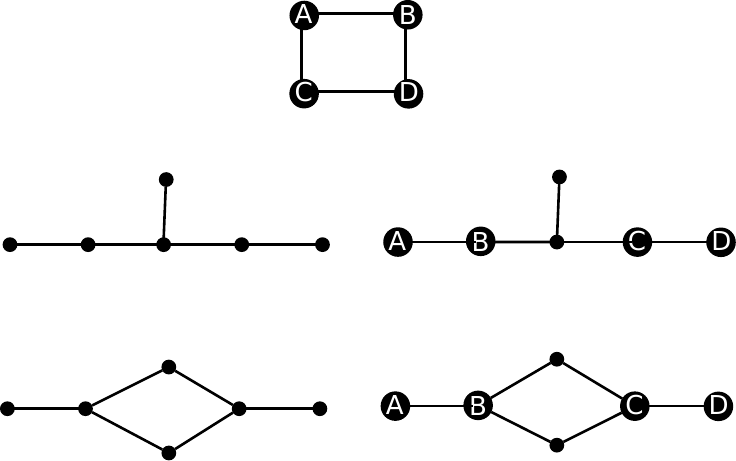}
  \caption{The cycle on 4 vertices, embedded into the two graphs on the left. The first embedding is not proper. The second embedding is proper.}
  \label{fig-proper}
\end{figure}

Given some target graph to embed into, there may not necessarily be a proper embedding. However, for some ``quasi-subgraph'' (defined below) of our target, there will be a proper embedding, which can be used to find an embedding into the target graph.

\begin{definition} \label{def:H_quasi}
Let $J$ and $J'$ be connected graphs. We say $J'$ is a \emph{quasi-subgraph} of $J$ if $J$ can be made isomorphic to $J'$ by applying any sequence of the following rules to $J$:
\begin{enumerate}
\item Delete a vertex in $V(J)$.
\item Delete an edge in $E(J)$.
\item Delete an edge $\{u, v\} \in E(J)$, add vertices $u', v'$ to $V(J)$, and add edges $\{u, u'\}, \{v, v'\}$ to $E(J)$.
\end{enumerate}
\end{definition}

\begin{figure}[H]
	\centering
  \includegraphics{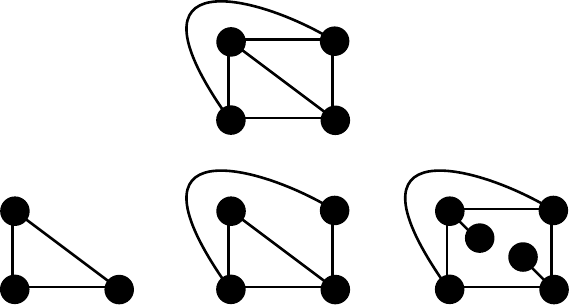}
  \caption{$K_4$ and three quasi-subgraphs of $K_4$.}
  \label{fig-quasi}
\end{figure}

We now show that by examining the quasi-subgraphs of our target graph, we can restrict our search to proper, pushing, non-contracting embeddings.

\begin{lemma} \label{lemma:H_proper}
There exists a proper, pushing, non-contracting $c_\OPT$-embedding of $G$ to some $(H^q, w^q)$, where $H^q$ is the subdivision of some quasi-subgraph of $H$, and $w^q : E(H^q) \rightarrow \mathbb{R}^{>0}$.
\end{lemma}
\begin{proof}
If $f_\OPT$ yields a non-contracting $c_\OPT$-embedding of $G$ to the line, then by a theorem of \cite{FellowsFLLSR09}, there exists a pushing, non-contracting $c_\OPT$-embedding of $G$ to the line. Since $G$ is connected, a line embedding must also be proper.
Therefore, in this case the claim of the lemma is true. For the rest of the proof, we shall assume $f_\OPT$ does not yield such an embedding.

Suppose that $f_\OPT$ is not proper. Let
\[
E^{\neg} = \{e \in E(H) : \exists z \in V(e_\OPT) \mbox{ such that $z$ is not proper w.r.t.~$f_\OPT$}\}
\]
and
\[
V^{\neg} = \{v \in V(H) : \mbox{ $v$ is not proper w.r.t.~$f_\OPT$}\}.
\]
Suppose that for some $e \in E(H)$ there exists $z_1, z_2 \in V(e_\OPT)$ such that $z_1$ and $z_2$ are both not proper w.r.t.~$f_\OPT$, and there exists $v \in V(G)$ such that $f_\OPT(v)$ is in the path in $e_\OPT$ from $z_1$ to $z_2$. Since $G$ is a connected graph, we therefore have that $f_\OPT$ embeds all of $V(G)$ to the path in $e_\OPT$ from $z_1$ to $z_2$. Therefore, $f_\OPT$ yields a non-contracting $c_\OPT$-embedding of $G$ to the line, contradicting our assumption above. Thus, for any vertex $z_3$ in the path in $e_\OPT$ from $z_1$ to $z_2$, we have that $z_3$ is not proper w.r.t.~$f_\OPT$.

Using the following procedure, we can modify $f_\OPT, (H_\OPT, w_\OPT)$ so that $f_\OPT$ is still a $c_\OPT$-embedding of $G$ to $(H_\OPT, w_\OPT$ and $|V^{\neg}| + |E^{\neg}|$ is reduced by at least one.

If $V^{\neg} \neq \emptyset$, choose $z \in V^{\neg}$.
Modify $H_\OPT$ by applying rule 2 of Definition \ref{def:H_quasi} to all $e \in E(H_\OPT)$ adjacent to $z$, and then rule 1 to $z$.
With these modifications, $H_\OPT$ is now a subdivision of the quasi-subgraph of $H$ found by applying rule 3 to all edges in $E(H)$ adjacent to $z$, and then rule 1 to all vertices in the component containing $z$.
Before these modification, for all $u, v \in V(G)$, there exists a path $P_{u, v} \in (H_\OPT, w_\OPT)$ from $f_\OPT(u)$ to $f_\OPT(v)$ such that $|P_{u, v}| = d_{(H_\OPT, w_\OPT)}(f_\OPT(u), f_\OPT(v))$ and $z \not\in V(P_{u, v})$. Therefore, after the modifications, a corresponding path $P'_{u, v}$ exists in $(H_\OPT, w_\OPT)$, with $|P_{u, v}| = |P'_{u, v}|$. Thus after these modifications, $f_\OPT$ is again a non-contracting, $c_\OPT$-embedding of $G$ to $(H_\OPT, w_\OPT)$, and $|V^{\neq}|$ is reduced by at least 1.

If $V^{\neg} = \emptyset$, choose $\{a, b\} \in E^{\neg}$.
Modify $H_\OPT$ by applying rule 1 of Definition \ref{def:H_quasi} to every $z \in Z = \{e \in V(e_\OPT) : e \mbox{ is not proper w.r.t.~$f_\OPT$}\}$. 
Before these modifications, for all $u, v \in V(G)$, there exists a path $P_{u, v} \in (H_\OPT, w_\OPT)$ from $f_\OPT(u)$ to $f_\OPT(v)$ such that $|P_{u, v}| = d_{(H_\OPT, w_\OPT)}(f_\OPT(u), f_\OPT(v))$ and $Z \cap V(P_{u, v}) = \emptyset$. Therefore, after the modifications, path $P_{u, v}$ exists in $(H_\OPT, w_\OPT)$. Therefore, there is a single connected component $C$ of $(H_\OPT, w_\OPT) \setminus Z$ such that $f_\OPT(V(G)) \subseteq (C, w_\OPT)$.
Modify $H_\OPT$ again by applying rule 1 to any $v \in V(H)$ such that $f_\OPT(v) \notin C$.
Thus $f_\OPT$ remains a non-contracting, $c_\OPT$-embedding of $G$ to $(H_\OPT, w_\OPT)$, and $|E^{\neg}|$ is reduced by one.

The sum $|V^{\neg}| + |E^{\neg}|$ is finite, and so a finite number of iterations of the procedure above will yield a proper, non-contracting $c_\OPT$ embedding of $G$ to $(H_\OPT, w_\OPT)$, where $H_\OPT$ is some quasi-subgraph of $H$.

Suppose $f_\OPT$ is a proper, non-contracting $c_\OPT$-embedding of $G$ to $(H_\OPT, w_\OPT)$, where $H_\OPT$ is some quasi-subgraph of $H$, and $f_\OPT$ is not pushing.
Therefore, there exists $a, b \in V(G)$ such that $a, b$ are consecutive w.r.t~some edge $e \in E(H_\OPT)$, and since $f_\OPT$ is non-contracting,
\[
d_{e_\OPT}(f_\OPT(a), f_\OPT(b)) > d_G(a, b).
\]
Modify $H_\OPT$ to replace the path in $e_\OPT$ between $f_\OPT(a), f_\OPT(b)$ with a single edge $\{f_\OPT(a), f_\OPT(b)\}$. Modify $w_\OPT$ so that
\[
w_\OPT(\{f_\OPT(a), f_\OPT(b)\}) = d_G(a, b).
\]
Then for all $u, v \in V(G)$, we either have that
\[
d_{(H_\OPT, w_\OPT)}(f_\OPT(u), f_\OPT(v))
\]
is unchanged by the modification to $e_\OPT$, or
\begin{align*}
d_{(H_\OPT, w_\OPT)}(f_\OPT(u), f_\OPT(v)) 
&= d_{(H_\OPT, w_\OPT)}(f_\OPT(u), f_\OPT(a)) \\
&\tab+ d_{(H_\OPT, w_\OPT)}(f_\OPT(b), f_\OPT(v)) + d_G(a, b) \\
&< d_{(H_\OPT, w_\OPT)}(f_\OPT(u), f_\OPT(v)),
\end{align*}
or 
\begin{align*}
d_{(H_\OPT, w_\OPT)}(f_\OPT(u), f_\OPT(v)) 
&= d_{(H_\OPT, w_\OPT)}(f_\OPT(v), f_\OPT(a)) \\
&\tab+ d_{(H_\OPT, w_\OPT)}(f_\OPT(b), f_\OPT(u)) + d_G(a, b) \\
&< d_{(H_\OPT, w_\OPT)}(f_\OPT(u), f_\OPT(v))
\end{align*}
If it is that case that 
\begin{align*}
d_{(H_\OPT, w_\OPT)}(f_\OPT(u), f_\OPT(v)) 
&= d_{(H_\OPT, w_\OPT)}(f_\OPT(u), f_\OPT(a)) \\
&\tab+ d_{(H_\OPT, w_\OPT)}(f_\OPT(b), f_\OPT(v)) + d_G(a, b) \\
&< d_{(H_\OPT, w_\OPT)}(f_\OPT(u), f_\OPT(v)),
\end{align*}
then by the triangle inequality we have that
\begin{align*}
d_{(H_\OPT, w_\OPT)}(f_\OPT(u), f_\OPT(a))\\
+ d_{(H_\OPT, w_\OPT)}(f_\OPT(b), f_\OPT(v)) + d_G(a, b) \\
&\geq d_G(u, a) + d_G(b, v) + d_G(a, b) \\
&\geq d_G(u, v)
\end{align*}
and therefore
\begin{align*}
d_G(u, v) &\leq d_{(H_\OPT, w_\OPT)}(f_\OPT(u), f_\OPT(v)) \\
&< d_{(H_\OPT, w_\OPT)}(f_\OPT(u), f_\OPT(v)) \\
&\leq c_\OPT \cdot d_G(u, v).
\end{align*}
Similarly, if it is the case that 
\begin{align*}
d_{(H_\OPT, w_\OPT)}(f_\OPT(u), f_\OPT(v)) 
&= d_{(H_\OPT, w_\OPT)}(f_\OPT(v), f_\OPT(a)) \\
&\tab+ d_{(H_\OPT, w_\OPT)}(f_\OPT(b), f_\OPT(u)) + d_G(a, b) \\
&< d_{(H_\OPT, w_\OPT)}(f_\OPT(u), f_\OPT(v))
\end{align*}
then we have that
\begin{align*}
d_G(u, v) &\leq d_{(H_\OPT, w_\OPT)}(f_\OPT(u), f_\OPT(v)) \\
&< d_{L(H_\OPT, w_\OPT)}(f_\OPT(u), f_\OPT(v)) \\
&\leq c_\OPT \cdot d_G(u, v).
\end{align*}
Therefore, after these modifications $f_\OPT$ remains a $c_\OPT$-embedding of $G$ to $(H_\OPT, w_\OPT)$.

By repeated modifications as described above, $f_\OPT$ can be modified until it is a proper, pushing, non-contracting $c_\OPT$-embedding of $G$ to $(H_\OPT, w_\OPT)$, with $H_\OPT$ a quasi-subgraph of $H$.
\end{proof}
Finally, we show the local density lemma. 

\begin{lemma}[Local Density] \label{lemma:local_density}
Let $f_\OPT$ be a non-contracting  $c_\OPT$-embedding of $G$ to some $(H^q, w^q)$, where $H^q$ is the subdivision of some quasi-subgraph of $H$, and $w^q : E(H^q) \rightarrow \mathbb{R}^{>0}$. Then for all $v \in V(G)$, for any $r \geq 0$,
\[
|\ball_{G}(v, r)| \leq 2r \cdot c_\OPT \cdot |E(H)|.
\]
\end{lemma}
\begin{proof}
Since $f_\OPT$ is a non-contracting $c_\OPT$-embedding, for any $u\in V(G)$ such that
\[
d_G(u, v) \leq r
\]
we have that
\begin{align*}
1 &\leq d_{(H_\OPT, w_\OPT)}(f_\OPT(u), f_\OPT(v)) \\
&\leq d_G(u, v) \\
&\leq c_\OPT \cdot d_{(H_\OPT, w_\OPT)}(f_\OPT(u), f_\OPT(v)) \\
&\leq c_\OPT \cdot r.
\end{align*}
Therefore, for each edge $e \in E(H)$, there are at most $2c_\OPT \cdot r$ vertices $x \in \ball_G(v, r)$ such that $f_\OPT(x) \in e_\OPT$, and therefore 
\[
|\ball_{G}(v, r)| \leq 2r \cdot c_\OPT \cdot |E(H^q)|.
\]
Thus, by Definition \ref{def:H_quasi}, 
\[
|\ball_{G}(v, r)| \leq 4r \cdot c_\OPT \cdot |E(H)|.
\]
\end{proof}

\section{An approximation algorithm for embedding into arbitrary graphs}\label{sec:H_approx}
In this section we give our approximation algorithm for embedding into 
arbitrary graph. In particular, we prove Theorem~\ref{thm:mainApprox}. 
By Lemma~\ref{lemma:H_proper} there is a proper, pushing $c_{\OPT}$ embedding of $G$ into a quasi-subgraph $H^q$ of $H$ with edge weight function $w^q$. Furthermore, by subdividing each edge of $H$ sufficiently many times, for any $\epsilon > 0$ any $c$-embedding of $G$ into $(H^q, w^q)$ can be turned into an $(c+\epsilon)$-embedding of $G$ into a subdivision of $H$.

The weighted quasi-subgraph $(H^q, w^q)$ of $H$ is a subdivision of a subgraph $H_{sub}$ of $H$. Since $H$ only has $2^{|E(H)+V(H)|}$ subgraphs our algorithm can guess $H_{sub}$. Thus, for the purposes of our approximation algorithm, it is sufficient to find an embedding of $G$ into a weighted subdivision $(H_{ALG}, w_{ALG})$ of $H_{sub}$ under the assumption that a  proper, pushing $c_{\OPT}$ embedding of $G$ into some weighted subdivision of $H_{sub}$ exists. Furthermore, any proper and pushing embedding is non-contracting and has contraction exactly equal to $1$. Such an embedding $f$ is a $c$-embedding if and only if for every edge 

\begin{align}\label{eqn:distDefnByEdge}
uv \in E(G), d_{(H,w)}(f(u), f(v)) \leq c.
\end{align}
Thus, to prove that our output embedding is a $c$-embedding (for some $c$) we will prove that it is proper, pushing and that~(\ref{eqn:distDefnByEdge}) is satisfied. Thus, the main technical result of this section is encapsulated in the following lemma.

\begin{lemma}\label{lem:mainApproxH}
There is an algorithm that takes as input a graph $G$ with $n$ vertices, a graph $H$ and an integer $c$, runs in time $2^h \cdot n^{\cO(1)}$ and either correctly concludes that there is no $c$-embedding of $G$ into a weighted subdivision of $H$, or produces a proper, pushing $c'$-embedding of $G$ into a weighted subdivision of a subgraph $H'$ of $H$, where $c_{\ALG} = \cO(c^{24}h^9)$. 
\end{lemma}

\paragraph{Definitions.} To prove Lemma~\ref{lem:mainApproxH} we need a few definitions. Throughout the section we will assume that there exists a weighted subdivision $(H_{\OPT}, w_{\OPT})$ and a $c$-embedding $f_{\OPT} : V(G) \rightarrow V(H_{\OPT})$. This embedding is unknown to the algorithm and will be used for analysis purposes only. Every edge $e=uv$ in $H$ corresponds to a path $P_e$ in $H_{\OPT}$ from $u$ to $v$. Based on the embedding $f_{\OPT} : V(G) \rightarrow V(H_{\OPT})$ we define the {\em embedding pattern} function $\hat{f}_{\OPT} : V(G) \rightarrow V(H) \cup E(H)$ as follows. For every vertex $v \in V(G)$ such that $f_{\OPT}$ maps $v$ to a vertex of $H_{\OPT}$ that is also a vertex of $H$, $\hat{f}_{\OPT}$ maps $v$ to the same vertex. In other words if $f_{\OPT}(v) = u$ for $u \in V(H)$, then $\hat{f}_{\OPT}(v) = u$. Otherwise $f_{\OPT}$ maps $v$ to a vertex $u$ on a path $P_e$ corresponding to an edge $e \in E(H)$. In this case we set $\hat{f}_{\OPT}(v) = e$. 

We will freely make use of the ``inverses'' of the functions $f_{\OPT}$ and $\hat{f}_{\OPT}$. For a vertex set $C \subseteq V(H_{\OPT})$ we define $f^{-1}_{\OPT} = \{v \in V(G) :  f_{\OPT}(v) \in C\}$. We will also naturally extend functions that act on elements of a universe to subsets of that universe. For example, for a set $F \subseteq E(H_{\OPT})$ we use $w_{\OPT}(F)$ to denote $\sum_{e \in F} w_{\OPT}(e)$. We further extend this convention to write $w_{\OPT}(P_e)$ instead of $w_{\OPT}(E(P_e))$ for a path (or a subgraph) of $H_{\OPT}$.
We extend the distance function to also work for distances between sets


Throughout the section we will use the following parameters, for now ignore the parenthesized comments to the definitions of the parameters, these are useful for remembering the purpose of the parameter when reading the proofs.
\begin{itemize}\setlength\itemsep{-.5mm}
\item $h = |E(H)|$ (the number of edges in $H$),
\item $c$ (the distortion of $f_{\OPT}$),
\item $\ell = 20c^3$ (long edge threshold)
\item $r = 5\ell h$ (half of covering radius)
\item $c_{ALG} = 64 \cdot 10^6 \cdot c^{24} (h+1)^9$ (distortion of output embedding) 
\end{itemize}

Using the parameter $\ell$ we classify the edges of $H$ into short and  long edges. An edge $e \in E(H)$ is called {\em short} if $w_{\OPT}(P_e)  \leq \ell$ and it is called {\em long} otherwise. The edge sets $E_{short}$ and $E_{long}$ denote the set of short and long edges in $H$ respectively. A {\em cluster} in $H$ is a connected component $C$ of the graph $H_{short} = (V(H), E_{short})$. We abuse notation and denote by $C$ both the connected component and its vertex set. The {\em long edge degree} of a cluster $C$ in $H$ is the number of long edges in $H$ incident to vertices in $C$. Here a long edges whose both endpoints are in $C$ is counted twice. A cluster $C$ of long edge degree at most $2$ is called {\em boring}, otherwise it is {\em interesting}. Most of the time when discussing clusters, we will be speaking of clusters in $H$. However we overload the meaning of the word cluster to mean something else for vertex sets of $G$. A {\em cluster in $G$} is a set $C$ such that there exists a cluster $C_H$ of $H$ such that $C = \{v \in V(G) ~:~ \hat{f}_{\OPT}(v) \in V(C_H) \cup E(C_H)\}$. Thus there is a one to one correspondence between clusters in $G$ and $H$.

The following lemma is often useful when considering embeddings into the line, or ``line-like structures''. We will need this lemma to analyze the parts of the graph $G$ that the embedding $f_{\OPT}$ maps to long edges of $H$

\begin{lemma} \label{lem:funnyIneq}
Suppose there exists a $c$-embedding $f_{\OPT}$ of $G$ into $(H_{\OPT}, w_{\OPT})$, and let $a$, $b$ and $v$ be vertices of $G$ such that $d_G(a, v) = d_G(b, v)$ and a shortest path from $a$ to $v$ in $G$ contains a vertex $w$ such that $d_H(f_{\OPT}(w), f_{\OPT}(b)) \leq c$. Then $d_G(a, b) \leq 2c$.
\end{lemma}

\begin{proof}
We have that 
$d_G(b, v) = d_G(a, v) = d_G(a, w) + d_G(w, s)$
And also that $d_G(b, v) \leq d_G(b, w) + d(w, s)$, but $f_{\OPT}$ is non-contracting so 
$d_G(b, w) \leq d_H(f_{\OPT}(w), f_{\OPT}(b)) + d(w, s) \leq c + d(w, s)$.
We conclude that $d_G(a, w) + d_G(w, s) \leq c + d(w, s)$, and cancelling $d(w, s)$ on both sides yields $d_G(a, w) \leq c$. Finally we have that $d_G(a, b) \leq d_G(a, w) +  d_G(w, b) \leq c + d_H(f_{\OPT}(w), f_{\OPT}(b)) \leq 2c$, concluding the proof.
\end{proof}

A {\em cluster-chain} is a sequence $C_1, e_1, C_2, e_2, \ldots, e_{t-1}, C_t$ such that the following conditions are satisfied. First, the $C_i$'s are distinct clusters in $H$, except that possibly $C_1 = C_t$. Second, $C_1$ and $C_t$ are interesting, while $C_2, \ldots C_{t-1}$ are boring. Finally, for every $i < t$ the edge $e_i$ is a long edge in $H$ connecting a vertex of $C_i$ to a vertex of $C_{i+1}$.

\subsection{Using Breadth First Search to Detect Interesting Clusters}
In this subsection we prove a lemma that is the main engine behind Lemma~\ref{lem:mainApproxH}. Once the main engine is set up, all we will need to complete the proof of Lemma~\ref{lem:mainApproxH} will be to complete the embedding by running the approximation algorithm for embedding into a line for each cluster-chain of $H$, and stitching these embeddings together.

Before stating the lemma we define what it means for a vertex set $F$ in $G$ to cover a cluster. We say that a vertex set $F \subseteq V(G)$ $r$-{\em covers} a cluster $C$ in $G$ if some vertex in $F$ is at distance at most $r$ from at least one vertex in $C$. A vertex set $F \subseteq V(G)$ covers a cluster $C$ in $H$ if $F$ covers the cluster $C_G$ corresponding to $C$ in $G$.

\begin{lemma}[Interesting Cluster Covering Lemma]\label{lem:interestingCluster}
There exists an algorithm that takes as input $G$, $H$ and $c$, runs in time $2^{h}n^{\cO(1)}$ and halts. If there exists a proper $c$-embedding $\hat{f}_{\OPT}$ from $G$ to a weighted subdivision of $H$, the algorithm outputs a family ${\cal F}$ such that $|{\cal F}| \leq 2^h$, every set $F \in {\cal F}$ has size at most $h$, and there exists an $F \in {\cal F}$ that $2r$-covers all interesting clusters of $H$.
\end{lemma}

Towards the proof of Lemma~\ref{lem:interestingCluster} we will design an algorithm that iteratively adds vertices to a set $F$. During the iteration the algorithm will make some non-deterministic steps, these steps will result in the algorithm returning a family of sets ${\cal F}$ rather than a single set $F$.

\subsection{The SEARCH algorithm}
We now describe a crucial subroutine of the algorithm of Lemma~\ref{lem:interestingCluster} that we call the SEARCH algorithm. The algorithm takes as input $G$, $c$, a set $F \subseteq V(G)$ and a vertex $v$. The algorithm explores the graph, starting from $v$ with the aim of finding a local structure in $G$ that on one hand, can not be embedded into the line with low distortion, while on the other hand is far away from $F$. It will either output {\sf fail}, meaning that the algorithm failed to find a structure not embeddable into the line, or {\sf success} together with a vertex $\hat{u}$, meaning that the algorithm succeeded to find a structure not embeddable into the line, and that $u$ is close to this structure. We begin with describing the algorithm, we will then prove a few lemmata describing the behavior of the algorithm. 

\paragraph{Description of the SEARCH algorithm}
The algorithm takes as input $G$, $c$, a set $F \subseteq V(G)$ and a vertex $v$. It performs a breadth first search (BFS) from $v$ in $G$. Let $X_1$, $X_2$, etc. be the BFS layers starting from $v$. In other words $X_i = \{x \in V(G) ~:~ d_G(v,x) = i\}$. The algorithm inspects the BFS layers $X_1, X_2, \ldots$ one by one in increasing order of $i$.

For $i < 2c^2$ the algorithm does nothing other than the BFS itself. For $i = 2c^2$ the algorithm proceeds as follows. It picks an arbitrary vertex $v_L \in X_i$ and picks another vertex $v_R \in X_i$ at distance at least $2c+1$ from $v_L$ in $G$. Such a vertex $v_R$ might not exist, in this case the algorithm proceeds without picking $v_R$.
The algorithm partitions $X_i$ into $X_i^L$ and $X_i^R$ in the following way. For every vertex $x \in X_i$, if $d_G(x, v_L) \leq 2c$ then $x$ is put into $X_i^L$. If $d_G(x, v_R) \leq 2c$ then $x$ is put into $X_i^R$. If some vertex $x \in X_i$ is put {\em both} into $X_i^L$ and in $X_i^R$, or {\em neither} into $X_i^L$ nor into $X_i^R$  the algorithm returns {\sf success} together with $\hat{u} = v$.

For $i > 2c^2$ the algorithm proceeds as follows. If any vertex in $X_i$ is at distance at most $r$ from any vertex in $F$ (in the graph $G$), the algorithm outputs {\sf fail} and halts. Otherwise, the algorithm partitions $X_i$ into $X_i^L$ and $X_i^R$. The vertex $x \in X_i$ is put into $X_i^L$ if $x$ has a neighbor in $X_{i-1}^L$ and into  $X_i^R$ if $x$ has a neighbor in $X_{i-1}^R$. Note that $x$ has at least one neighbor in $X_{i-1}$, and so $x$ will be put into at least one of the sets  $X_{i}^L$ or $X_i^R$. If $x$ is put into both sets  $X_{i}^L$ and $X_i^R$, the algorithm outputs {\sf success} with $\hat{u} = x$ and halts. If $|X_i^L| > 2c^2$ or if two vertices in $X_i^L$ have distance at least $2c+1$ from each other in $G$ the algorithm picks a vertex $x \in X_i^L$ and returns {\sf success} with $\hat{u} = x$. Similarly, if $|X_i^R| > 2c^2$ or if two vertices in $X_i^R$ have distance at least $2c+1$ from each other in $G$ the algorithm picks a vertex $x \in X_i^R$ and returns {\sf success} with $\hat{u} = x$. If the BFS stops, (i.e $X_i = \emptyset$), the algorithm outputs {\sf fail}.

\paragraph{Properties of the SEARCH algorithm.}
We will only give guarantees on the behavior of the SEARCH algorithm provided that there exists a $c$-embedding $f_{\OPT}$ of $G$ into $(H_{\OPT}, w_{\OPT})$, and that $v$ is at distance at least $4c^2$  from every cluster $C$ in $G$. Therefore {\em within this subsection, $v$ refers to the vertex SEARCH is started from, and all lemmas assume that these two conditions are satisfied.} In this case we have that $\hat{f}_{\OPT}(v) = e$ for a long edge $e \in E(H)$. The edge $e$ belongs to a unique cluster-chain $C_1, e_1, C_2, e_2, \ldots, e_{t-1}, C_t$. For some $q \leq t-1$ we have that $e_q = e$. 

The vertex $v$ splits the cluster-chain in two parts, $C_1, e_1, C_2, e_2, \ldots, C_{q-1}$ and $C_q, e_{q+1}, \ldots, C_{t}$, we may think of these as the ``left'' and the ``right'' part of the chain. The edge $e = e_q$ is ``split down the middle'' in the following sense, the path $P_{e}$ is divided in two parts $P_e^L$, defined as the sub-path of $P_e$ from the endpoint in $C_{q-1}$ to $f_{\OPT}(v)$, and $P_e^R$, defined as the sub-path of $P_e$ from $f_{\OPT}(v)$ to the endpoint in $C_{q}$. We now define the left and the right part of the chain:
\begin{eqnarray*}
L & = & \left(\bigcup_{i \leq q-1} \hat{f}_{\OPT}^{-1}(C_i) \cup \hat{f}_{\OPT}^{-1}(e_i) \right) \cup \hat{f}_{\OPT}^{-1}(C_q) \cup f_{\OPT}^{-1}(P_e^L) \\
R & = & f_{\OPT}^{-1}(P_e^R) \cup \left(\bigcup_{i \geq q+1} \hat{f}_{\OPT}^{-1}(C_i) \cup \hat{f}_{\OPT}^{-1}(e_i) \right) 
\end{eqnarray*}
Note that $L$ and $R$ are vertex sets in $G$. The sets $L$ and $R$ intersect only in $v$, unless $C_1 = C_t$, in which case both $L$ and $R$ contain $\hat{f}_{\OPT}^{-1}(C_1) = \hat{f}_{\OPT}^{-1}(C_t)$. No other vertices are common to $L$ and $R$. We define $\zeta = L \cup R$ to be the set of all vertices of $G$ on the chain. 

We will say that the {\em left and right side of the search met in iteration $i$} if SEARCH put some vertex $x \in X_i$ both in $X_i^L$ and in $X_i^R$. In this case the algorithm outputs {\sf success} and halts in this iteration. We also say that the algorithm {\em left-succeeded (right-succeeded)} if it output {\sf success} with $\hat{u} \in X_i^L$ ($\hat{u} \in X_i^R$) for some $i$.

The focus of our analysis is on how SEARCH explores $\zeta$. We say that  SEARCH {\em leaves $\zeta$ in iteration $i$} if $i$ is the lowest number $X_i \setminus \zeta \neq \emptyset$. We say that {\em the inner part of $\zeta$} is $\zeta_{inner} = \zeta \setminus (\hat{f}_{\OPT}^{-1}(C_1) \cup \hat{f}_{\OPT}^{-1}(C_t))$. We say that SEARCH {\em leaves the inner part of $\zeta$ in iteration $i$} if $i$ is the lowest number such that $X_i \setminus \zeta_{inner} \neq \emptyset$. The next lemma shows that that $X_i^L$ and $X_i^R$ correctly classify the vertices of $\zeta$ into $L$ and $R$ as long as SEARCH has not yet left $\zeta$, and as long as  the left and right side of the search have not met.

\begin{lemma}\label{lem:classifyStrong}
if $i \geq 2c^2$, and SEARCH does not halt or leave $\zeta$ in any iteration $j \leq i$, and $v_L \in L$ then
$$X_i^L = X_i \cap L \mbox{ and } X_i^R = X_i \cap R.$$
If $v_L \in R$ then
$$X_i^L = X_i \cap R \mbox{ and } X_i^R = X_i \cap L.$$
\end{lemma}

\begin{proof}
We show the lemma when $v_L \in L$, the case when $v_L \in R$ is symmetric. We first show the statement of the lemma for $i = 2c^2$, and start by proving that $v_R \in R$. Suppose not, then either the shortest path from $v_R$ to $v$ in $G$ contains a vertex $w$ such that $d_{H_{\OPT},w_{\OPT}}(f_{\OPT}(w), f_{\OPT}(v_L)) \leq c$ or the shortest path from $v_L$ to $v$ in $G$ contains a vertex $w$ such that $d_{H_{\OPT},w_{\OPT}}(f_{\OPT}(w), f_{\OPT}(v_R)) \leq c$. In either case, Lemma~\ref{lem:funnyIneq} shows that $d_G(v_L, v_R) \leq 2c$, contradicting the choice of $v_R$. We conclude that $v_R \notin L$ and therefore that $v_R \in R$ (if it exists). 

We have that $X_i \cap R \neq \emptyset$ because the embedding $f_{\OPT}$ is proper. Furthermore we have that 
\[
d_{H_{\OPT},w_{\OPT}}( f_{\OPT}(v_L), f_{\OPT}(v)) \geq 2c^2
\]
 and that for any $x \in X_i \cap R$ we have that 
\[ 
 d_{H_{\OPT},w_{\OPT}}( f_{\OPT}(x), f_{\OPT}(v)) \geq 2c^2.
\]
Thus 
\[
d_{H_{\OPT},w_{\OPT}}( f_{\OPT}(x), f_{\OPT}(v_L)) \geq 4c^2
\]
implying that $d_G(x, v_L) \geq 4c$. Thus, the SEARCH algorithm does indeed pick a vertex $v_R \in R$.
 
Now, for any vertex $v_L' \in X_i \cap L$ we have that either the shortest path from $v_L'$ to $v$ in $G$ contains a vertex $w$ such that $d_{H_{\OPT},w_{\OPT}}(f_{\OPT}(w), f_{\OPT}(v_L)) \leq c$ or the shortest path from $v_L$ to $v$ in $G$ contains a vertex $w$ such that $d_{H_{\OPT},w_{\OPT}}(f_{\OPT}(w), f_{\OPT}(v_L')) \leq c$. In either case, Lemma~\ref{lem:funnyIneq} shows that $d_G(v_L, v_L') \leq 2c$, implying that $v_L' \in X_i^L$. An identical argument shows that every $v_R' \in X_i \cap R$ is in $X_i^R$. Since $X_i^L$ and $X_i^R$ form a partition of $X_i$ this proves the statement of the lemma for $i = 2c^2$.

Suppose now that the statement of the lemma holds for every $i' < i$ (with $i' \geq 2c^2$), we prove the lemma for $i$. If SEARCH halts in iteration $i$ there is nothing to prove, so assume that SEARCH does not halt in iteration $i$. Then the left and right side of the search did not meet in iteration $i$. This means that every vertex $u$ in $X_i$ either has a neighbor $u'$ in $X_{i-1}^L = X_i \cap L$ or in $X_{i-1}^R = X_i \cap R$, but not both.

If $u' \in X_{i-1}^L$ then SEARCH puts $u$ into $X_i^L$. Furthermore we have that $u'$ is in $L$, and 
\[
d_{H_{\OPT},w_{\OPT}}(f_{\OPT}(u'),f_{\OPT}(v)) \geq 2c^2,
\]
while $d_{H_{\OPT},w_{\OPT}}(f_{\OPT}(u'),f_{\OPT}(u)) \leq c$. Thus we conclude that $u \in L$.
By an identical argument, if $u' \in X_{i-1}^R$ then SEARCH puts $u$ into $X_i^R$ and $u \in R$. Since both $X_i^L, X_i^R$ and $X_i \cap L, X_i \cap R$ form partitions of $X_i$ the lemma  follows.
\end{proof}

\begin{lemma}\label{lem:leaveThenHalt}
If SEARCH leaves the inner chain in iteration $i$, then before reaching iteration $i + c + \ell \cdot h + 4c^2$, SEARCH either succeeds or fails by finding a vertex within distance $r$ from $F$.
\end{lemma}

\begin{proof}
In iteration $i$, SEARCH visits a vertex $u \in X_i \setminus \zeta_{inner}$, $u$ has a neighbor $u' \in \zeta_{inner} \cap X_{i-1}$. We have that $u'$ is either in $L$ or in $R$, without loss of generality we have that $u' \in L$. Since $u \notin \zeta$ it follows that $\hat{f}_{\OPT}(u') = e_1$. Since $u'u \in E(G)$ and $u \notin \zeta_{inner}$ it follows that $d_{H_{\OPT},w_{\OPT}}(f_{\OPT}(u'), C_1) \leq c$. In $H_{\OPT}$ the distance between all vertices of $C_1$ is at most $\ell \cdot h$. Since $f_{\OPT}$ is non-contracting a BFS (and thus SEARCH) will visit all of $f_{\OPT}^{-1}(C_1)$ by iteration $i + c + \ell \cdot h$.

At this point, either the left and right side of the search have already met (in which case the algorithm succeeded), the algorithm encountered a vertex at distance at most $r$ from $F$ (in which case it failed), or it leaves $\zeta$ within iteration $i + c + \ell \cdot h + 1$. Since the embedding is proper, we have that for some iteration $j \leq i + c + \ell \cdot h + 4c^2$, the search visits a vertex $x \in X_j$ such that $\hat{f}_{\OPT}(x) = e^x$ for $e^x \neq e_1$ and $4c^2 \geq d_{H_{\OPT},w_{\OPT}}(f_{\OPT}(x), C_1) \leq 4c^2 + c$. Let $j$ be the first iteration such that this event occurs, and $x$ and $e_x$ as defined above for this iteration $j$. We remark that technically $e^x$ might not be an edge different from $e_1$ but rather the other endpoint of $e_1$. This does not affect the proof other than in notation, so we will treat $e^x$ and $e_1$ as different edges.

We claim that unless SEARCH already has halted, in iteration $j$, $X_j^L$ contains a vertex $y'$ at distance more than $2c$ from $x$, making SEARCH succeed. This is all that remains to prove in order to prove the statement of the lemma. 

Since $C_1$ is an interesting cluster, $C_1$ is incident to at least one more long edge $e^y$ distinct from $e_1$ and $e^x$. Again, technically $e_y$ could be the other endpoint of the $e^x$ or $e_1$, however this does not affect the proof and thus we treat them as separate edges. Let $y$ be a vertex in $G$ such that $\hat{f}_{\OPT}(y) = e^y$ and $4c^2 \geq d_{H_{\OPT},w_{\OPT}}(f_{\OPT}(y), C_1) \leq 4c^2 + c$. We have that $d_G(u', y) \leq d_{H_{\OPT},w_{\OPT}}(f_{\OPT}(u'), f_{\OPT}(y)) \leq c + \ell \cdot h + 4c^2$. By the choice of $x$ we have that $y$ is not discovered by SEARCH before $x$ is.

The subgraph of $H_{\OPT}$ corresponding to the cluster $C_1$ and the sub-path of $P_{e^y}$ from $C_1$ to $f_{\OPT}(y)$ is connected, and therefore there is an index $j' \geq j$ such that $j' \leq j+c$ such that $X_{j'}^L$ contains a vertex $y^*$ such that $\hat{f}_{\OPT}(y^*) \in C_1 \cup \{e^y\}$. We have that $d_{H_{\OPT},w_{\OPT}}(f_{\OPT}(y^*), f_{\OPT}(x)) \geq 4c^2$, hence $d_{G}(y^*, x) \geq 4c$. However there exists a predecessor $y'$ of $y^*$ in the BFS such that $y' \in X_j^L$ and $d_G(y, y') \leq c$. The triangle inequality yields that $d_G(y', x) > 4c - c > 2c$ and the statement follows.
\end{proof}

Finally we show that whenever SEARCH succeeds, the vertex it outputs is near a cluster.

\begin{lemma}\label{lem:closeCluster}
If SEARCH outputs {\sf success} and a vertex $\hat{u}$, then there exists a cluster $C$ in $G$ such that $d_G(\hat{u}, C) \leq 2c + \ell \cdot h + 4c^2 \leq 2\ell h$.
\end{lemma}

\begin{proof}
SEARCH can succeed either because the left and right side of the search meet, or because SEARCH left-succeeds or because it right-succeeds. If the left and right side of the search meet in iteration $i$, it means that in this or one of the previous iterations the search has visited a vertex $u$ such that $d_{H_{\OPT}, w_{\OPT}}(f_{\OPT}(u), C_1) \leq c$. By Lemma~\ref{lem:leaveThenHalt} it follows that SEARCH halts within $2c + \ell \cdot h + 4c^2$ iterations and outputs a vertex $\hat{u}$ within distance $2c + \ell \cdot h + 4c^2$ from $C_1$.

Suppose now that SEARCH left-succeeds in iteration $i$, and assume for contradiction that $d_G(\hat{u}, C) > 2c + \ell \cdot h + 4c^2$ for every cluster $C$ in $G$. If the output vertex  $\hat{u}$ is not in $\zeta_{inner}$, then Lemma~\ref{lem:leaveThenHalt} again yields that $d_G(\hat{u}, \hat{f}_{\OPT}^{-1}(C_1)) \leq c + \ell \cdot h + 4c^2$. Therefore, assume that $\hat{u} \in \zeta_{inner}$.
In this case there is an edge $e_p$ on the cluster-chain of $v$ such that $\hat{f}_{\OPT}(\hat{u}) = e_p$. The edge $e_p$ connects the clusters $C_{p}$ and $C_{p+1}$. Since $d_G(\hat{u}, \hat{f}_{\OPT}^{-1}(C_{p} \cup C_{p+1})) \geq 10c^2$ we have that $d_{H_{\OPT}, w_{\OPT}}(f_{\OPT}(\hat{u}), C_{p} \cup C_{p+1}) \geq 10c^2$.
Let $u'$ be the predecessor of $\hat{u}$ in the BFS, we have that $u' \in X_{i-1}^L$. Since SEARCH did not succeed in iteration $i-1$ we have that $d_G(u', u'') \leq 2c$ for every $u'' \in X_{i-1}^L$.
Since every vertex in $X_i^L$ has a predecessor in $X_{i-1}^L$ we conclude that every vertex in $X_i^L$ is at distance at most $2c+2$ from $\hat{u}$ in $G$. Thus, every vertex in $f_{\OPT}(X_i^L)$ is at distance at most $2c^2 + 2c$ in $(H_{\OPT}, w_{\OPT})$ from $f_{\OPT}(\hat{u})$. Most importantly $\hat{f}_{\OPT}(u'') = e_p$ for every $u'' \in X_i^L$.
Therefore, for every pair of vertices $u''$ and $u^*$ in $X_i^L$, either the shortest path from $u''$ to $v$ contains a vertex $w$ such that $d_{H_{\OPT}, w_{\OPT}}(f_{\OPT}(w), f_{\OPT}(u^*)) \leq c$ or the shortest path from $u^*$ to $v$ contains a vertex $w$ such that $d_{H_{\OPT}, w_{\OPT}}(f_{\OPT}(w), f_{\OPT}(u'')) \leq c$. It follows from Lemma~\ref{lem:funnyIneq} that $d_G(u^*, u'') \leq c$. Since this holds for every pair of vertices $u''$ and $u^*$ in $X_i^L$ this contradicts that the algorithm left-succeeded in iteration $i$. The proof if the algorithm right-succeeded is symmetric.
\end{proof}

\paragraph{The COVER algorithm}
We are now almost in position to prove Lemma~\ref{lem:interestingCluster}. We begin by describing the COVER algorithm, and then prove that it satisfies the conditions of Lemma~\ref{lem:interestingCluster}. We will describe the COVER algorithm as a {\em non-deterministic} algorithm that takes as input $G$, $H$ and $c$, runs in time polynomial time, and outputs a single set vertex set $F \subseteq V(G)$ of size at most $h$.  If there exists a proper $c$-embedding $\hat{f}_{\OPT}$ from $G$ to a weighted subdivision of $H$, then in at least one of the computation paths of the algorithm, the output set $F$ $2r$-covers all interesting clusters of $H$. The algorithm will use only $h$ non-deterministic bits. By defining ${\cal F}$ to be the family containing all sets $F$ output by the computation paths of COVER, the family ${\cal F}$ satisfies the conditions of Lemma~\ref{lem:interestingCluster}.

The COVER algorithm proceeds as follows, given $G$, $H$ and $c$, it initializes $F = \emptyset$. It then proceeds in stages. In stage $i$ the algorithm loops over all choices of a vertex $v \in V(G) \setminus \ball_G(F,r)$, and runs SEARCH on $G$, starting from $v$, with the set $F$. If SEARCH fails for all choices of $v$ the COVER algorithm terminates and outputs $F$. Otherwise, let $v$ be the first vertex that made SEARCH succeed in stage $i$, and let $\hat{u}$ be the vertex output by SEARCH. The algorithm makes a non-deterministic choice: in one computation path $v$ is added to $F$, in the other computation path $\hat{u}$ is added to $F$. Then the algorithm proceeds to stage $i+1$. If the algorithm reaches stage $h+1$ it terminates without outputting any set. This concludes the description of the algorithm.

\begin{proof}[Proof of the Interesting Cluster Covering Lemma (Lemma~\ref{lem:interestingCluster})]
Each stage of the COVER algorithm ends when SEARCH started from a vertex $v$ succeeds and outputs a vertex $\hat{u}$. The entire analysis of SEARCH is only valid if $v$ is at distance at least $4c^2$ from every cluster $C$ in $G$. The non-deterministic guess of the COVER algorithm is whether this assumption is valid; i.e whether $G_G(v,C) \geq 4c^2$ for every cluster $C$. We proceed analyzing the computation path where the non-deterministic guess is correct.

If $v$ is at distance at least $4c^2$ from every cluster $C$ in $G$, the COVER algorithm adds $\hat{u}$ to $F$, otherwise COVER adds $v$ to $F$. In either case the vertex added to $F$ is at least at distance $r+1$ from every other vertex in $F$. Furthermore, if COVER adds $v$ to $F$ then $v$ is within distance $4c^2$ from some cluster $C$ in $G$. On the other hand, if COVER adds $\hat{u}$ to $F$ then, by Lemma~\ref{lem:closeCluster} there exists a cluster $C$ in $G$ such that $d_G(\hat{u}, C) \leq 2c + \ell \cdot h + 4c^2 \leq 2\ell h$.

Since every pair of vertices in a cluster $C$ are at distance at most $\ell \cdot h$ apart in $G$, every vertex in $F$ is at distance at most $2c + \ell \cdot h + 4c^2 \leq 2\ell h$ away from a cluster, and every pair of vertices in $F$ are at distance at least $r \geq 5\ell h$ apart, we have that in the computation path that makes the correct guesses the algorithm terminates and outputs a set $F$ of size at most $h$ before reaching stage $h+1$.

To complete the proof we need to show that every cluster in $C$ is $2r$-covered  by $F$. Suppose not, and consider an un-covered cluster $C$ in the last stage of the COVER algorithm. In this stage, SEARCH failed when starting from every vertex $v$ of $G$. Let $v$ be a vertex at distance exactly $4c^2$ from $C$ such that such that $\hat{f}_{\OPT}(v)$ is a long edge incident to the cluster in $H$ corresponding to $C$. By Lemma~\ref{lem:leaveThenHalt}, starting SEARCH from $v$ will result in the algorithm halting within $4c^2 + c + \ell \cdot h \ 4c^2 \leq 2\ell h$ iterations. Furthermore, if the algorithm does not succeed (which it does not, since this was the last stage of COVER), it finds a vertex $u$ at distance at most $r$ from $F$. But then $d_G(F, C) \leq d_G(F, u) + d_G(u, C) \leq r + 2\ell h + 4c^2 \leq 2r$, completing the proof.
\end{proof}

\subsection{STITCHing Together Approximate Line Embeddings}
We now describe the STITCH algorithm. This algorithm takes as input $G$, $H$, $c$ and $F \subseteq V(G)$, runs in polynomial time and halts. We will prove that if there exists a $c$-embedding $f_{\OPT}$ of $G$ into a weighted subdivision $(H_{\OPT}, w_{\OPT})$ of $H$ such that all $F$ $2r$-covers all interesting clusters of $G$, the algorithm produces a $c_{\ALG}$-embedding $f_{\ALG}$ of $G$ into a weighted subdivision $(H_{\ALG}, w_{\ALG})$ of a subgraph $H'$ of $H$. Throughout this section we will assume that such an embedding $f_{\OPT}$ exists. 

The STITCH algorithm starts by setting $R = 4r$, $\Delta = 4r$ and then proceeds as follows. As long as there are two vertices $v_i$ and $v_j$ in $F$ such that $2R \leq d_G(u, v) \leq  2R + \Delta$, the algorithm increases $R$ to $R + \Delta$. Note that this process will stop after at most ${|F| \choose 2}$ iterations, and therefore when it terminates we have $R \leq 4r \cdot h^2 \leq 400c^3 h^3$. Define $B = \ball_G(F,R)$, and ${\cal B}$ to be the family of connected components of $G[B]$. Notice that the previous process ensures that for any $B_1, B_2 \in {\cal B}$ we have $d_G(B_1, B_2) \geq \Delta$. Notice further that for every interesting cluster $C$ in $H$ we have that $\ball_G(\hat{f}_{\OPT}^{-1}(C), r) \subseteq G$. 

We now classify the connected components of $G - B$.  A component $Z$ of $G - B$ is called {\em deep} if it contains at least one vertex at distance(in $G$) at least $\frac{\Delta}{2}$ from $F$, and it is {\em shallow} otherwise. The shallow components are easy to handle because they only contain vertices close to $F$.

\begin{lemma}\label{lem:shallowGrave}
For every shallow component $Z$ of $G - B$, there is at most one connected component $B_1 \in {\cal B}$ that contains neighbors of $Z$
\end{lemma}

\begin{proof}
Suppose not, then there are two vertices $v_1$ and $v_2$ in $Z$ that are neighbors, such that the closest vertex in $B$ to $v_1$   is in $B_1$ while the closest vertex in $B$ to $v_2$   is in $B_2$, for distinct components $B_1$ and $B_2 \in {\cal B}$. The distance from $v_1$ to $B_1$ is at most $\frac{\Delta}{2}-1$, the distance from $v_2$ to $B_2$ is at most $\frac{\Delta}{2}-1$, and hence, by the triangle inequality, the distance between $B_1$ and $B_2$ is at most $\Delta - 1$, contadicting the choice of $R$.
\end{proof}
  
The next sequence of lemmas allows us to handle deep components. We say that a component  $Z$ in $G - B$ {\em lies on} the cluster-chain $\chi = C_1, e_1, \ldots, C_t$ if 
$$Z \subseteq \left( \bigcup_{i \leq t} \hat{f}_{\OPT}^{-1}(C_i) \cup \hat{f}_{\OPT}^{-1}(e_i) \right) \setminus \hat{f}_{\OPT}^{-1}(C_1 \cup C_t).$$

\begin{lemma}\label{lem:onChain} Every component $Z$ of $G - B$ lies on some cluster-chain. \end{lemma}

\begin{proof}
$Z$ does not contain any vertices in interesting clusters, or even within distance $r$ of interesting clusters. No two vertices that (a) are at least $c$ from all interesting clusters and (b) are mapped by $f_{\OPT}$ on different cluster-chains can be adjacent, because the distance between their $f_{\OPT}$ images in $H$ is at least $2c$. The lemma follows.
\end{proof}

\begin{lemma}\label{lem:noOnSame} No two deep components $Z_1$, $Z_2$ of $G - B$ can lie on the same cluster-chain $\chi$ \end{lemma}

\begin{proof}
Suppose to such deep components exist. Because $\Delta = 4r$ and $r > \ell \cdot h$, and every cluster of $G$ has at most $\ell \cdot h$ vertices, it follows that $Z_1$ contains a vertex $v_1$ such that the distance from $v_1$ to any cluster in $G$ is at least $2c^2$ and $d_G(v_1, B) \geq \Delta/4$. Thus $\hat{f}_{\OPT}(v_1) = e_i$ for an edge $e_i$ on the cluster chain $\chi$. By an identical argument there is a vertex $v_2$ in $Z_2$ such that the distance from $v_2$ to any cluster in $G$ is at least $2c^2$ and $d(v_2, B) \geq \Delta/4$. Thus $\hat{f}_{\OPT}(v_2) = e_j$ for an edge $e_j$ on the cluster chain $\chi$.

Without loss of generality $i \leq j$ and if $i = j$ then $f_{\OPT}(v_1)$ is closer than $f_{\OPT}(v_2)$ to the endpoint of $P_{e_i}$ that lies in $C_i$.

The graph $H_{\OPT} \setminus f_{\OPT}(\{v_1, v_2\})$ has two connected components, one that contains $C_1$ and $C_t$, and one that does not. Consider the connected component $\zeta$ that does not. Because the embedding $f_{\OPT}$ is proper, $G[f_{\OPT}^{-1}(\zeta)]$ contains a path $P$ with one endpoint within distance at most $c^2$ from $v_1$, and the other within distance at most $c^2$ from $v_2$. Since $d(\{v_1 \cup v_2\},B) \geq \Delta/4$ we have that one endpoint of $P$ is in $Z_1$ and the other is in $Z_2$. But any path from $Z_1$ to $Z_2$ (and in particular $P$) must contain a vertex from $B$. This implies that $\zeta \cap B \neq \emptyset$. 

This yields a contradiction: we have that the component $B_i$ of $G[B]$ that has non-empty intersection with $\zeta$ also has non-empty intersection with an interesting cluster. It follows that $B_i$ contains a vertex within distance at most $c$ from either $v_1$ or $v_2$, contradicting the choice of $\{v_1,v_2\}$.
\end{proof}

\begin{lemma}\label{lem:deepCompEmbed} There is a polynomial time algorithm that given $G$, $B$ and a component $Z$ of $G - B$ computes an embedding of $Z$ components of $G - B$ into the line with distortion at most $(\ell \cdot h \cdot c)^4$. Furthermore, all vertices in $Z$ with neighbors outside $Z$ are mapped by this embedding within distance $(\ell \cdot h \cdot c)^6$ from the end-points. 
\end{lemma}

\begin{proof}
Let $Z$ be a component of $G - B$. By Lemma~\ref{lem:onChain}, $Z$ lies on a cluster-chain $\chi = C_1, e_1, \ldots, C_t$. Define a following total ordering of the vertices in $Z$: 
If $\hat{f}_{\OPT}(a) \in C_i \cup \{e_i\}$ and $\hat{f}_{\OPT}(b) \in C_j \cup \{e_j\}$ and $i < j$, then $a$ comes before $b$.
If $\hat{f}_{\OPT}(a) \in C_i$ and $\hat{f}_{\OPT}(b) = e_i$ then $a$ comes before $b$.
If $\hat{f}_{\OPT}(a) = \hat{f}_{\OPT}(b) = e_i$ and ${f}_{\OPT}(a)$ is closer than ${f}_{\OPT}(b)$ to $C_{i-1}$, then $a$ comes before $b$.
If $\hat{f}_{\OPT}(a) \in C_i$ and $\hat{f}_{\OPT}(b) \in  C_i$ break ties arbitrarily.

At most $\ell \cdot h$ vertices are mapped to any boring cluster $C_i$, and the distance between any two vertices in the same boring cluster in $H_{\OPT}$ is at most $\ell \cdot h$. Thus the distance (in $G$) between any two consecutive vertices in this ordering is at most $\ell \cdot h \cdot c$. The number of vertices appearing in the ordering between the two endpoints of an edge is at most $\ell \cdot h$ (all the vertices of a boring cluster). Thus, if the ordering is turned into a pushing, non-contracting embedding into the line, the distortion of this embedding is at most $(\ell \cdot h)^2 \cdot c$. Using the known polynomial time approximation algorithm for embedding into the line~\cite{BadoiuDGRRRS05} we can find an embedding of $Z$ into the line with distortion at most $(\ell \cdot h \cdot c)^4$ in polynomial time.

Because $B$ is a union of at most $h$ balls, it follows that at most $c^2 \cdot h^2$ vertices in $Z$ have neighbors in $G$, and that all of these vertices are among the $\ell \cdot h$ first or last ones in the above ordering. Since any two low distortion embeddings of a metric space into the line map the same vertices close to the end-points, it follows that all vertices in $Z$ with neighbors outside $Z$ are mapped by this embedding within distance $(\ell \cdot h \cdot c)^6$ from the end-points. 
\end{proof}

The STITCH algorithm builds the graph $H'$ as follows. Every vertex of $H'$ corresponds to a connected component $B \in {\cal B}$. Every deep component $Z$ of $G - B$ corresponds to an edge between the (at most two) sets $B_1$ and $B_2 \in {\cal B}$ that have non-empty intersection with $N_G(Z)$. Note that the graph $H'$ is a multi-graph because it may have multiple edges and self loops. However, since each set $B \in {\cal B}$ has a connected image in $H$ under $\hat{f}$, Lemmata~\ref{lem:onChain} and~\ref{lem:noOnSame} imply that $H'$ is a topological subgraph of $H$. Hence any weighted subdivision of $H'$ is a weighted subdivision of a subgraph of $H$.

The STITCH algorithm uses Lemma~\ref{lem:deepCompEmbed} to compute embeddings of each deep connected component $Z$ of $G \setminus B$. Further, for each component $B_i \in {\cal B}$ the algorithm computes the set $B_i^\star$ which contains $B_i$, as well as the vertex sets of all shallow connected components whose neighborhood is in $B_i$. By Lemma~\ref{lem:shallowGrave} the $B_i^\star$'s together with the deep components of $G - B$ form a partition of $V(G)$.

What we would like to do is to map each set $B_i^\star$ onto the vertex of $H'$ that it corresponds to, and map each deep connected component $Z$ of $G-B$ onto the edge of $H'$ that it corresponds to. When mapping $Z$ onto the edge of $H$ we use the computed embedding of $Z$ into the line, and subdivide this edge appropriately. 

The reason this does not work directly is that we may not map all the vertices of $B_i^\star$ onto the single vertex $v_i$ in $H'$ that corresponds to $B_i$. Instead, STITCH picks one of the edges incident to $v_i$, sub-divides the edge an appropriate number of times, and maps all the vertices of $B_i^\star$ onto the newly created vertices on this edge. The order in which the vertices of $B_i^\star$ are mapped onto the edge is chosen arbitrarily, however all of these vertices are mapped closer to $v_i$ than any vertices of the deep component $Z$ that is mapped onto the edge. This concludes the construction of $H_{\ALG}$ and 
$f_{\ALG}$. The STITCH algorithm defines a weight function $w_{\ALG}$ on the edges of $H_{\ALG}$, such that the embedding is pushing and non-contracting.

\begin{lemma}\label{lem:appxSuccess}
$f_{\ALG}$ is a $c_{\ALG}$-embedding of $G$ into $(H_{\ALG}, w_{\ALG})$.
\end{lemma}

\begin{proof}
It suffices to show that the distance in $(H_{\ALG}, w_{\ALG})$ between the image of two endpoints of an edge $uv \in E(G)$ is never more than $c_{\ALG}$. To that end, the main observation is every $B_i \in {\cal B}$ is the union of at most $h$ balls of radius $R$. Every vertex of $B_i^\star$ is within distance $\Delta/2 \leq R$ from some vertex in $B_i$. Hence, for any two vertices $a, b \in B_i^\star$ we have that $d_G(a, b) \leq 2 \cdot R \cdot (h+1)$. Thus, by Lemma~\ref{lemma:local_density} we have that $|B_i| \leq 4Rh(h+1)$. Therefore, for any $B_i^\star$, the embedding $f_{\ALG}$ embeds $B_i^\star$ on a path of length at most $8R^2h(h+1)^2 \leq 8R^2(h+1)^3$ in $(H_{\ALG}, w_{\ALG})$.

Every edge with both endpoints in $B_i^\star$ is therefore stretched at most $8R^2(h+1)^3$ by $f_{\ALG}$. By Lemma~\ref{lem:deepCompEmbed}, every edge with both endpoints in a deep component $Z$ of $G \setminus B$ is stretched at most $(\ell \cdot h \cdot c)^4$. Furthermore, by Lemma~\ref{lem:deepCompEmbed}, any edge with one endpoint in  $B_i^\star$  and the other in $Z$ is stretched at most $8R^2(h+1)^3 + (\ell \cdot h \cdot c)^6$. Hence every edge is stretched at most $c_{\ALG}$ completing the proof.
\end{proof}

\subsection{The Approximation Algorithm}
We are now ready to prove Lemma~\ref{lem:mainApproxH}, for convenicence we re-state the lemma here.

\medskip
\noindent {\bf Lemma~\ref{lem:mainApproxH}.} {\em There is an algorithm that takes as input a graph $G$ with $n$ vertices, a graph $H$ and an integer $c$, runs in time $2^h \cdot n^{\cO(1)}$ and either correctly concludes that there is no $c$-embedding of $G$ into a weighted subdivision of $H$, or produces a proper, pushing $c_{\ALG}$-embedding of $G$ into a weighted subdivision of a subgraph $H'$ of $H$, where $c_{\ALG} = 64 \cdot 10^6 \cdot c^{24} (h+1)^9$.}

\begin{proof}
The algorithm runs the COVER algorithm, to produce a collection ${\cal F}$, such that  $|{\cal F}| \leq 2^h$, every set in ${\cal F}$ has size at most $h$, and such that if $G$ has a $c$-embedding $f_{\OPT}$ of into a weighted subdivision of $H$, then some $F \in {\cal F}$ $2r$-covers all interesting clusters (of $f_{\OPT}$) in $G$. For each $F \in {\cal F}$ the algorithm runs the STITCH algorithm, which takes polynomial time. If STITCH outputs a $c_{\ALG}$-embedding of $G$ into a weighted subdivision of a subgraph $H'$ of $H$, the algorithm returns this embedding.

By Lemma~\ref{lem:appxSuccess}, for the choice of $F \in {\cal F}$ that $2r$-covers all interesting clusters, the STITCH algorithm does output a $c_{\ALG}$-embedding of $G$ into a weighted subdivision of a subgraph $H'$ of $H$. This concludes the proof.
\end{proof}

The discussion prior to the statement of Lemma~\ref{lem:mainApproxH} immediately implies that Lemma~\ref{lem:mainApproxH} is sufficient to give an approximation algorithm for finding a low distortion (not necessarily pushing, proper or non-contracting) embedding $G$ into a (unweighted) subdivision of $H$. The only overhead of the algorithm is the guessing of the subgraph $H_{sub}$ of $H$, this incurs an additional factor of $2^{|V(H)| + |E(H)|} \leq 4^h$ in the running time, yielding the following theorem.

\begin{theorem}\label{thm:mainApprox}
There exists a $8^hn^{\cO(1)}$ time algorithm that takes as input an $n$-vertex graph $G$, a graph $H$ on $h$ vertices, and an integer $c$, and either correctly concludes that there is no $c$-embedding from $G$ to a subdivision of $H$, or produces a $c_{\ALG}$-embedding of $G$ into a subdivision of $H$, with $c_{\ALG} \leq  64 \cdot 10^6 \cdot c^{24} (h+1)^9$.
\end{theorem}

Finally, we remark that at a cost of a potentially higher running time in terms of $h$, one may replace the $(h+1)^9$ factor with $c^9$. If $c \geq h+1$ we have that $c_{\ALG} \leq 64 \cdot 10^6 \cdot c^{33}$. On the other hand, if $c \leq h+1$ we may run the algorithm of Theorem~\ref{theorem:fpt_H} in time $f(H)n^{\cO(1)}$ instead and solve the problem optimally.

\section{A FPT algorithm for embedding into arbitrary graphs}\label{sec:H_fpt}
In this section we design our FPT algorithm for embedding into 
arbitrary graph. In particular we show the following result in this section.

\begin{theorem} \label{theorem:fpt_H}
Given an integer $c>0$ and graphs $G$ and $H$, it is possible in time $f(H, c) \cdot n^{\cO(1)}$ to either find a non-contracting $c$-embedding of $G$ into some subdivision of $H$, or correctly determine that no such embedding exists.
\end{theorem}

The proof of Theorem \ref{theorem:fpt_H} will come at the end of the section. Using Lemma \ref{lemma:H_proper}, for the rest of this section we shall assume w.l.o.g.~that $f_\OPT$ is a proper, pushing, non-contracting $c_\OPT$-embedding of $G$ into $(H_\OPT, w_\OPT)$, where $H_\OPT$ is a subdivision of $H^q$, some quasi-subgraph of $H$, and $w_\OPT : E(H^q) \rightarrow \mathbb{R}^{>0}$.

\begin{definition}
For any $e \in E(H^q)$, we say $e$ is \emph{short} if $w^q(e) \leq 16(c_\OPT)^4$. Otherwise, $e$ is \emph{long}.
\end{definition}

Based on this definition of short and long edges, we define the following notions of clusters in $H^q$.

\begin{definition}
Let $\mathcal{C}$ the set of connected components of $H^q \setminus \{e \in E(H) : \mbox{$e$ is long}\}$. We say that $C \in \mathcal{C}$ is an \emph{interesting cluster of $H^q$} if there exist at least 3 paths leaving $C$ in $H^q$. Let 
\[
\mathcal{C}^{\geq3} = \{C \in \mathcal{C} : C \mbox{is an interesting clusters of } H^q\}
\]
and let
\[
\mathcal{C}^{<3} = \mathcal{C} \setminus \mathcal{C}^{\geq3}.
\]
\end{definition}

\begin{definition}
For each connected component $C$ of $H^q \setminus \mathcal{C}^{\geq3}$, we say $C$ is a \emph{path cluster of $H^q$}.
Let $\mathcal{P}$ be the set of path clusters of $H^q$.
\end{definition}

The following lemma describes the 3 categories these path clusters may fall into.

\begin{lemma}
For all $P \in \mathcal{P}$, one of the following cases holds:
\begin{description}
\item{\textbf{Case 1.}} There exists $k>0$ and a sequence 
\[
e_1, C_1, \ldots, e_k, C_k
\]
such that 
\begin{enumerate}
\setlength{\itemsep}{-2pt}
\item $e_1, \ldots, e_k$ are long edges of $H^q$. 
\item $C_1, \ldots, C_k \in \mathcal{C}^{<3}$.
\item $P = e_1 \cup C_1 \cup \ldots \cup e_k \cup C_k$.
\item There exists $C \in \mathcal{C}^{\geq3}$ such that $C \cap e_1 \neq \emptyset$.
\end{enumerate}
\item{\textbf{Case 2.}} There exists $k>0$ and a sequence 
\[
e_1, C_1, \ldots, e_k, C_k, e_{k+1}
\]
such that
\begin{enumerate}
\setlength{\itemsep}{-2pt}
\item $e_1, \ldots, e_{k+1}$ are long edges of $H^q$. 
\item $C_1, \ldots, C_k \in \mathcal{C}^{<3}$.
\item $P = e_1 \cup C_1 \cup \ldots \cup e_k \cup C_k$.
\item There exists $C \in \mathcal{C}^{\geq3}$ such that $C \cap e_1 \neq \emptyset$ and for all $C' \in \mathcal{C}^{\geq3}$, $C' \cap e_{k+1} = \emptyset$.
\end{enumerate}
\item{\textbf{Case 3.}} There exists $k>0$ and a sequence 
\[
e_1, C_1, \ldots, e_k, C_k, e_{k+1}
\]
such that
\begin{enumerate}
\setlength{\itemsep}{-2pt}
\item $e_1, \ldots, e_{k+1}$ are long edges of $H^q$. 
\item $C_1, \ldots, C_k \in \mathcal{C}^{<3}$.
\item $P = e_1 \cup C_1 \cup \ldots \cup e_k \cup C_k$.
\item There exists $C, C' \in \mathcal{C}^{\geq3}$ such that $C \cap e_1 \neq \emptyset$ and $C' \cap e_{k+1} \neq \emptyset$.
\end{enumerate}
\end{description}
\end{lemma}
\begin{proof}
Let $H'$ be the graph which results from contracting all short edges of $H^q$. In $H'$, each $C \in \mathcal{C}^{<3}$ is expressed as a vertex of degree 1 or 2, and each $C' \in \mathcal{C}^{\geq3}$ as a vertex of degree 3 or more. If all vertices of degree 3 are removed, the remaining components must be paths. Since $H^q$ is a connected graph, each path component was connected to some vertex of degree 3 through one or both of the endpoints of the path.
\end{proof}

To find an embedding, it will be necessary to partition the vertices of $G$ into those which must be embedded near an interesting cluster, and those which do not. The following definition defines which vertices these will be. The theorem and lemma following the definition show that finding these vertices is a tractable problem.

\begin{definition} \label{def:important_vertex}
Let $v\in V(G)$ and $\alpha\geq 1$.
We say that $v$ is \emph{$\alpha$-interesting} if the metric space $(\ball_G(v, \alpha), d_G)$ does not admit a $c_\OPT$-embedding into the line.
\end{definition}

\begin{theorem}[Fellows \etal~\cite{FellowsFLLSR09}]\label{thm:weighted_line}
There exists an algorithm which given a weighted graph $\Gamma$, with weights in $\{1,\ldots,W\}$, and some $c\geq 1$, decides whether $\Gamma$ admits a $c$-embedding into the line in time $\cO(n (cW)^4 (2c+1)^{2cW})$.
\end{theorem}

\begin{lemma}[Importance is tractable] \label{lemma:importance_tractable}
There exists an algorithm which given $v\in V(G)$ and $\alpha\geq 0$, decides whether $v$ is $\alpha$-interesting, in time $\cO(n(c_\OPT 2 \alpha)^4 (2c_\OPT+1)^{4c_\OPT \alpha})$. 
\end{lemma}

\begin{proof}
Let $\Gamma$ be the complete weighted graph with $V(\Gamma)=\ball_G(v,\alpha)$, and such that for all $\{x,y\}\in {V(\Gamma) \choose 2}$, the length of $\{x,y\}$ is set to $d_G(x,y)$.
By the triangle inequality, it follows that the maximum edge length in $\Gamma$ is at most $2\alpha$.
Thus, by Theorem \ref{thm:weighted_line} we can decide whether $\Gamma$ admits a $c_\OPT$-embedding into the line in time $\cO(n(c_\OPT 2 \alpha)^4 (2c_\OPT+1)^{4c_\OPT \alpha})$, as required.
\end{proof}

Our algorithm will proceed by finding partial embeddings of $G$ into the interesting and path clusters. Later, these partial embeddings will be ``stitched'' together to form a complete embedding. To aid in the stitching process, we define a notion of compatibility between partial embeddings on quasi-subgraphs of $H$.

\begin{definition}
Let $H_1, H_2$ be quasi-subgraphs of $H$ such that there exists $\{a, b\} \in E(H_1) \cap E(H_2)$ so that $a$ is a leaf node in $H_1$, and $b$ is a leaf node in $H_2$. 
Let $f_1$ and $f_2$ be pushing, non-contracting $c_\OPT$-embeddings of subgraphs $G_1, G_2$ of $G$ into $(H^q_1, w_1)$ and $(H^q_2, w_2)$. We say \emph{$f_1$ and $f_2$ are compatible on $\{a, b\}$} if 
\begin{enumerate}
\setlength{\itemsep}{-2pt}
\item For all $v \in V(G_1) \cap V(G_2)$, $f_1(v) \in \sub_{H_1}(\{a, b\})$ and $f_2(v) \in \sub_{H_2}(\{a, b\})$
\item For all $u, v \in V(G_1) \cap V(G_2)$, we have $u, v$ are consecutive w.r.t.~$\{a, b\}$ if and only if $u, v$ are consecutive w.r.t.~$\{a, b\}$.
\item There exists $u' \in V(G_1) \cap V(G_2)$ such that $f_1(u') = a$.
\item There exists $v' \in V(G_1) \cap V(G_2)$ such that $f_1(v') = b$.
\end{enumerate}
If $f_1$ and $f_2$ are compatible on $\{a, b\}$, then we can combine $f_1, f_2$ in the following way: 
\begin{enumerate}
\setlength{\itemsep}{-2pt}
\item For every $u \in V(G_1) \cap V(G_2)$, let $f_1(u) = f_2(u)$.
\item For any $u, v \in V(G_1) \cap V(G_2)$ consecutive w.r.t.~$\{a, b\}$, replace the shortest path in $G_1$ from $f_1(u)$ to $f_1(v)$ and the shortest path in $G_2$ from $f_2(u)$ to $f_2(v)$ with a single edge of weight $d_G(u, v)$. All other edges have their weight from $w_1$ or $w_2$.
\end{enumerate}
\end{definition}

The a parameter $\Delta$ will appear in several places within the algorithm. We set the value of $\Delta$ now.

\begin{definition} \label{def:Delta}
Let
\[
\Delta = \diam(H^q) + 8 \cdot (c_\OPT)^4.
\]
\end{definition}

Our algorithm will make use of two sub-procedures, CLUSTER and PATH.

\subsection{CLUSTER algorithm}
The CLUSTER algorithm will find embeddings restricted to the interesting clusters of $H^q$. Let $S \subseteq V(G)$, $C$ a subgraph of $H^q$.

\begin{definition}
Let $e_1, \ldots, e_{|E(C)|}$ be some fixed ordering of $E(C)$, and for each $i \in \{1, \ldots, |E(C)|\}$, let $e_i = \{h_{i,1}, h_{i,2}\}$.
\end{definition}

\begin{definition}
We say $f_C, (C', w')$ is a \emph{solution of $\cluster(S, C)$} if $C'$ is a subdivision of $C$, $w' : E(C) \rightarrow \mathbb{R}^{>0}$, and $f_C : S \rightarrow (C', w')$ such that:
\begin{enumerate}
\setlength{\itemsep}{-2pt}
\item For all $u, v \in S$,
\[
d_G(u, v) \leq d_{(C', w')}(f_C(u), f_C(v)) \leq c_\OPT \cdot d_G(u, v).
\]
\item For all $u, v \in S$, if $u$ and $v$ are consecutive then
\[
d_{(C', w')}(f_C(u), f_C(v)) = d_G(u, v).
\]
\end{enumerate}
\end{definition}

\begin{definition}
A \emph{configuration of $S, C$} consists of the following:
\begin{enumerate}
\setlength{\itemsep}{-2pt}
\item A partition $E_1, \ldots, E_{|E(C)|}$ of $S$.
\item An ordering $O_i = o_{i,1}, \ldots, o_{i,|E_i|}$ of each $E_i$. Let 
\[
|O_i| = \sum_{j = 1}^{|E_i|-1} d_G(o_{i,j}, o_{i,j+1}).
\]
\item Let 
\[
\chi(o_{i,1}) = h_{i,1},
\]
\[
\chi(o_{i,|E_i|}) = h_{i,|E_i|},
\]
and
\[
\Omega = \cup_{i=1}^{|E(C)|} \{o_{i,1}, o_{i,|E_i|}\}.
\]
For each $x, y \in \Omega$, the configuration has a simple path $P_{x, y}$ in $C$ from $\chi(x)$ to $\chi(y)$.
\end{enumerate}
\end{definition}

The following algorithm will be used to generate solutions to $\cluster(S, C)$:
\begin{description}
\item{\textbf{Step 1.}} For each choice of configuration of $S, C$:
\begin{description}
\item {\textbf{Step 1.1.}}Minimize $\sum_{i=1}^{|E(C)|} \alpha_i + \beta_i$ subject to the following constraints:
\begin{itemize}
\item For all $i \in \{1, \ldots, |E(C)|$, 
\[
\alpha_i \geq 0
\]
and 
\[
\beta_i \geq 0.
\]
For all $z \in \Omega$, if $z = o_{i,1}$ for some $i \in \{1 \ldots, |E(C)|\}$, then let
\[
\omega(z) = \alpha_i,
\]
and if $z = o_{i,|E_i|}$ for some $i$, then let
\[
\omega(z) = \beta_i.
\]
\item For all $a, b \in V(H)$, for each path $P$ from $a$ to $b$, let
\[
|P| = \sum_{e_i \in E(P)} (\alpha_i + \beta_i + |O_i|).
\]
For all $x, y \in \Omega$,
\[
\ell_{x,y} = \omega(x) + \omega(y) + |P_{x,y}| \geq d_G(x, y)
\]
and for all other simple paths $P$ from $\chi(x)$ to $\chi(y)$,
\[
\omega(x) + \omega(y) + |P| \geq \ell_{x,y}.
\]
\end{itemize}
\item{\textbf{Step 1.2.}} Define the subdivision $C'$. For each edge $e_i \in E(C)$, if 
\[
\alpha_i \neq 0 \mbox{ and } \beta_i \neq 0
\]
then subdivide $e_i$ $|E_i|$ many times.
Otherwise, if 
\[
\alpha_i + beta_i \neq 0
\]
then subdivide $e_i$ $|E_i| - 1$ many times.
Otherwise, subdivide $e_i$ $|E_i| - 2$ many times.
\item{\textbf{Step 1.3.}} Define an embedding $f_C$. 
For all $i \in \{1, \ldots, |E(C)|\}$, let $e_i = \{a, b\}$, $a < b$. If $\alpha_i = 0$ then let 
\[
f_C(o_{i,1}) = a,
\]
otherwise let $f_C(o_{i,1})$ be the first vertex in the subdivision of $\{a, b\}$. 
If $\beta_i = 0$ then let 
\[
f_C(o_{i,|E_i|}) = b,
\]
otherwise let $f_C(o_{i,1})$ be the last vertex in the subdivision of $\{a, b\}$.
For each $j \in \{2, \ldots, |E_i|-1\}$, let $f_C(o_{i,j})$ be the fist vertex on the path in the subdivision of $e_i$ from $f_C(o_{i,j-1})$ and $b$.
\item{\textbf{Step 1.4.}} Define the weight function $w'$. For all $e_i = \{a, b\} \in E(C)$, $a < b$, if $\alpha_i \neq 0$ let 
\[
w'(\{a, f_C(o_{i,1})\}) = \alpha_i,
\]
and if $\beta_i \neq 0$ let
\[
w'(\{a_k, f_C(o_{i,|E_i|})\}) = \beta_i,
\]
and for all $j \in \{1, \ldots, |E_i|-1\}$, let
\[
w'(\{f_C(o_{i,j}), f_C(o_{i,j+1})\}) = d_G(o_{i,j}, o_{i,j+1}).
\]
\end{description}
\end{description}

\begin{lemma} \label{lemma:cluster_alg_bound}
The above algorithm finds $\cO(|E(C)|^{|S|} \cdot |S|! \cdot (|V(C)| - 2)!)$ solutions to $\cluster(S, C)$.
\end{lemma}
\begin{proof}
The algorithm finds one solution for each configuration of $S, C$.
There are no more than
\[
|E(C)|^{|S|}
\]
possible partitions $E_1, \ldots, E_{|E(C)|}$.
Given $E_1, \ldots, E_{|E(C)|}$, there are no more than
\[
|S|!
\]
possible orderings $O_1, \ldots, O_{|E(C)|}$.
For any $x, y \in V(C)$, there are no more than
\[
(|V(C)| - 2)! 
\]
simple paths between $x$ and $y$.
Therefore, there are no more than
\[
|E(C)|^{|S|} \cdot |S|! \cdot (|V(C)| - 2)! 
\]
configurations of $S, C$.
\end{proof}

\subsection{PATH algorithm}

The PATH algorithm will find embeddings restricted to the path clusters of $H^q$.

Let $P$ be a path cluster of $H^q$ such that
\[
P = e_1, C_1, e_2, C_2, \ldots, e_j, C_j
\]
or 
\[
P = e_1, C_1, e_2, C_2, \ldots, e_j, C_j, e_{j+1}.
\]
Let $X \subset V(G)$, $S=s_1,\ldots,s_{4(c_\OPT)^2 + 1}$, $T=t_1,\ldots,t_{4(c_\OPT)^2 + 1}$ or $T = \emptyset$, be sequences of vertices such that $V(S) \cup V(T) \subseteq V(G) \setminus X$ and $V(S) \cap V(T) = \emptyset$.
Let 
\[
Z = X \cup V(S) \cup V(T)
\]
and
\[
Z^\ell = \{v \in X \cup V(S) \cup V(T) : |\ball_G(v, (4(c_\OPT)^2 + 1) \cdot c_\OPT))| \leq (4(c_\OPT)^2 + 1) \cdot (c_\OPT)^2\}.
\]

Here we adapt the idea of feasible partial embeddings from \cite{FellowsFLLSR09} to our needs.  

\begin{definition}
A \emph{partial embedding of $A \subseteq Z^\ell$} is a bijective function 
\[
f : A \rightarrow \{0, \ldots, 4(c_\OPT)^2 + 1\}.
\]
Let
\begin{enumerate}
\setlength{\itemsep}{-1pt}
\item $f_e$ be the embedding of $A_f$ into $(e', w')$ derived in the following way:
\begin{enumerate}
\setlength{\itemsep}{-1pt}
\item Let $e=\{a, b\} \in E(H^q)$, $a < b$.
\item Let $e'$ be the subdivision of $e$ with $4(c_\OPT)^2+1$ vertices. Let $v_1, v_2, \ldots, v_{4(c_\OPT)^2+1}$ be the sequence of vertices encountered when traversing $e'$ from $a$ to $b$.
\item For all $a \in A_f$, let $f_e(a) = v_{f(a)}$.
\item For all $i \in \{1, \ldots, 4(c_\OPT)^2+1\}$, let $w'(\{v_i, v_{i+1}\}) = d_G(f^{-1}(i), f^{-1}(i+1))$.
\end{enumerate}
\item $A^L_f = f^{-1}(\{0, \ldots, 2(c_\OPT)^2\})$.
\item $A^R_L = f^{-1}(\{2(c_\OPT)^2 + 1, \ldots, 4(c_\OPT)^2+1\})$.
\item $L(A)$ is the union of the vertex sets of all connected components $C$ of $Z^\ell \setminus A$ such that $C$ has a neighbor in $A^L_f$, and the union of the vertex sets of all connected components $C'$ of $Z^\ell \setminus A$ such that $\ball_G(C', c_\OPT) \cap C \neq \emptyset$.
\item $R(A)$ is the union of the vertex sets of all connected components $C$ of $Z^\ell \setminus A$ such that $C$ has a neighbor in $A^R_f$, and the union of the vertex sets of all connected components $C'$ of $Z^\ell \setminus A$ such that $\ball_G(C', c_\OPT) \cap C \neq \emptyset$.
\end{enumerate}
\end{definition}

\begin{definition}
A partial embedding $f$ of $A \subseteq Z^\ell$ is called \emph{feasible} if  
\begin{enumerate}
\setlength{\itemsep}{-2pt}
\item $f_e$ is a proper, pushing, non-contracting $c_\OPT$-embedding of $(A_f, d_G)$ into $(e', w')$.
\item $L(f) \cap R(f) = \emptyset$.
\item $\ball_G(f^{-1}(2(c_\OPT)^2), c_\OPT)$ is in $A$.
\item For all $i \in \{0, \ldots, 4(c_\OPT)^2\}$,
\[
d_G(f^{-1}(i), f^{-1}(i+1)) \leq c_\OPT.
\]
\end{enumerate}
\end{definition}

\begin{lemma} \label{lemma:feasible_vertex_bound}
The number of feasible partial embeddings of $Z^\ell$ is $n \cdot (c_\OPT)^{\cO(c_\OPT)}$.
\end{lemma}
\begin{proof}
For every feasible partial embedding starting with $v_0$, there exists a sequence $v_0, v_1, \ldots, v_{4(c_\OPT)^2 + 1}$ such that for all $i \in \{0, \ldots, v_{4(c_\OPT)^2}\}$ we have
\[
d_G(v_i, v_{i+1}) \leq c_\OPT,
\]
and therefore for all $i \in \{1, \ldots, v_{4(c_\OPT)^2+1}\}$ we have
\[
d_G(v_0, v_i) \leq (4(c_\OPT)^2 + 1) \cdot c_\OPT.
\]
Since for all $v \in Z^\ell$, we have that 
\[
|\ball_G(v, (4(c_\OPT)^2 + 1) \cdot c_\OPT))| \leq (4(c_\OPT)^2 + 1) \cdot (c_\OPT)^2,
\]
and so there are at most $(4(c_\OPT)^2 + 1) \cdot (c_\OPT)^2$ vertices which can be in any partial embedding starting with $v_0$.
Therefore, the number of possible such sequences is 
\[
{(4(c_\OPT)^2 + 1) \cdot (c_\OPT)^2 \choose 4(c_\OPT)^2} \leq (c_\OPT)^{\cO(c_\OPT)}
\]
for each $v_0 \in Z^\ell$.
\end{proof}

\begin{definition}
Let $f$ and $g$ be feasible partial embeddings of $Z^\ell$, with domains $A_f$ and $A_g$. We say $g$ succeeds $f$ if 
\begin{enumerate}
\setlength{\itemsep}{-2pt}
\item $A_f \setminus \{f^{-1}(0)\} = A_g \setminus \{g^{-1}(4(c_\OPT)^2 + 1)\} = A_f \cap A_g$.
\item For every $a \in A_f \cap a_g$, $f(a) = g(a) + 1$.
\item $\{g^{-1}(4(c_\OPT)^2 + 1)\} \subseteq R(f)$.
\item $\{f^{-1}(0)\} \subseteq L(g)$
\end{enumerate}
\end{definition}

\begin{definition}
A feasible partial embedding of $W \subseteq Z$ is a 3-tuple $F=(f, r, R)$ such that 
\begin{enumerate}
\setlength{\itemsep}{-2pt}
\item $f$ is a feasible partial embedding of $Z^\ell$
\item $r \in \{0, 1, \ldots, j\}$.
\item If $r=j$ then $R = \emptyset$.
\item If $r < j$ and $e_{j+1} \in P$ then $R$ is a solution to
\[
\cluster(\ball_{Z}(A_f, \Delta) \cap (A_f \cup R(f)), e_{r+1} \cup C_{r+1} \cup e_{r+2})
\]
such that $R$ and $f$ are compatible w.r.t.~$e_{r+1}$.
\item If $r < j$ and $e_{j+1} \notin P$ then $R$ is a solution to
\[
\cluster(\ball_{Z}(A_f, \Delta) \cap (A_f \cup R(f)), e_{r+1} \cup C_{r+1})
\]
such that $R$ and $f$ are compatible w.r.t.~$e_{r+1}$.
\end{enumerate}
\end{definition}

\begin{lemma} \label{lemma:feasible_limit}
There are at most $n \cdot (c_\OPT)^{\cO(c_\OPT)} \cdot 2|E(H)| \cdot |E(H)|^{\cO(|E(H)|^2)}$ feasible partial embeddings of $Z$.
\end{lemma}
\begin{proof}
By Lemma \ref{lemma:feasible_vertex_bound}, we have that there are 
\[
n \cdot (c_\OPT)^{\cO(c_\OPT)}
\]
many feasible partial embeddings of $Z^\ell$.
Since $P \subseteq H^q$, we have that $k \leq 2|E(H)|$, and thus 
\[
r \leq 2|E(H)|.
\]
Each of $C_1, \ldots, C_k$ are subgraphs of $H^q$, and $Z$ is a subgraph of $G$. Therefore, by Lemma \ref{lemma:cluster_alg_bound}, $R$ is one of at most 
\[
\cO(2|E(H)|^{2c_\OPT \cdot |E(H)|^2} \cdot (2c_\OPT \cdot |E(H)|^2)! \cdot (2|V(H)| - 2)!) = |E(H)|^{\cO(|E(H)|^2)}.
\] 
solutions.

Therefore, there are at most 
\[
n \cdot (c_\OPT)^{\cO(c_\OPT)} \cdot 2|E(H)| \cdot |E(H)|^{\cO(|E(H)|^2)}
\]
feasible partial embeddings of $Z$.
\end{proof}

\begin{definition}
Let $F_1=(f_1, r_1, R_1)$, $F_2=(f_2, r_2, R_2)$ be two feasible partial embeddings of $Z$. We say $F_2$ succeeds $F_1$ if either of the following conditions are met:
\begin{enumerate}
\setlength{\itemsep}{-2pt}
\item $r_1 = r_2$ and $f_2$ succeeds $f_1$.
\item $r_2 = r_1+1$, $e_{r_2} \in P$, and $f_2$, $R_1$ are compatible on $e_{r_2}$.
\end{enumerate}
\end{definition}

\begin{definition} \label{def:path_cluster_embedding}
Let $F_1, \ldots, F_t$ be a sequence of feasible partial embeddings of $Z$ such that $L(F_1) = \emptyset$, $R(F_t) = \emptyset$, and for all $i \in \{2, \ldots, t\}$, we have that $F_i = (f_i, r_i, R_i)$, and $F_i$ succeeds $F_{i-1}$.
Let $f_P, (P', w')$ be the embedding of $Z$ derived from the sequence in the following way:
\begin{enumerate}
\setlength{\itemsep}{-2pt}
\item{\textbf{Step 1.}} While the $r$ values do not change, proceed through the sequence in order, while building $f_P, P', w'$ in the obvious way, so that $f_P$ is a pushing embedding.
\item{\textbf{Step 2.}} If a value $i$ is reached so that $r_i > r_{i-1}$, use $R_i$ to find the subdivision, embedding, and weights for $C_i$. Advance to edge $e_{r_i}$ where $R_i$ left off, and return to Step 1.
\end{enumerate}
\end{definition}

\begin{lemma} \label{lemma:feasible_partial_embeddings_fopt}
Let $P$ be a path cluster of $H^q$, let 
\[
I^P = \{v \in V(H^q) : v \mbox{ connects P to some interesting cluster} \}
\]
and let 
\[
Z_P = \{v \in V(G) : f_\OPT(v) \in \sub_{(H_\OPT, w_\OPT)}(P) \mbox{ and } d_{(H_\OPT, w_\OPT)}(v, I^P) \geq (4(c_\OPT)^2 + 1) \cdot c_\OPT\}
\]
For any path cluster of $H^q$, there is a sequence $F_1, f_2, \ldots, F_k$ of feasible partial embeddings of $Z_P$ such that $L(F_1) = \emptyset$, $R(F_k) = \emptyset$, and for all $i \in \{2, \ldots, k\}$, we have that $F_i$ succeeds $F_{i-1}$.
\end{lemma}
\begin{proof}
Since $P$ is a path cluster of $H^q$, we have that $1 \leq |I^P|
\leq 2$.
Choose $s \in I^P$, and orient each long edge of $P$ away from $s$. If both ends of $P$ connect to $P$, forming a cycle, choose a clockwise or counter-clockwise direction in which to orient the long edges.

Let 
\[
Z^\ell = \{v \in Z : |\ball_G(v, (4(c_\OPT)^2 + 1) \cdot c_\OPT))| \leq (4(c_\OPT)^2 + 1) \cdot (c_\OPT)^2\}.
\]

For each long edge $e=\{e_1, e_2\}$ of $P$, let $Z_e$ be the sequence of vertices such that
\[
V(Z_e) = \{v \in Z : f_\OPT(v) \in e_\OPT \mbox{ and } d_{(H_\OPT, w_\OPT)}(f_\OPT(v), \{e_1, e_2\}) \geq (4(c_\OPT)^2+1) \cdot c_\OPT\}
\]
and $Z_e$ has the order imposed on $V(Z_E)$ by $f_\OPT$, traversing $e$ along the orientation.
For any $z \in V(Z_e)$, we have that 
\[
|\ball_G(v, (4(c_\OPT)^2 + 1) \cdot c_\OPT))| \leq (4(c_\OPT)^2 + 1) \cdot (c_\OPT)^2,
\]
since $f_\OPT$ is a $c_\OPT$-embedding, and so $f_\OPT$ must embed $\ball_G(v, (4(c_\OPT)^2 + 1) \cdot c_\OPT))$ within $e$.
Therefore, 
\[
V(Z_e) \subseteq Z^\ell.
\]

Let $Z_i$ be the contiguous subsequence of $Z_e$ starting from the $i$-th vertex of $Z_e$ such that $|Z_i| = 4(c_\OPT)^2 + 1$.
Let $g_i$ be a function such that for any $z_j \in V(Z_i)$,
\[
g_i(z_j) = j.
\]
Thus $g_i$ is a partial embedding of $V(Z_i)$.
Since $f_\OPT$ is a proper, pushing, non-contracting $c_\OPT$-embedding, $g_i$ is a feasible partial embedding of $V(Z_i)$, and for all $i \in \{2, \ldots, |V(Z_i)| - 4(c_\OPT)^2+1\}$, we have that $g_i$ succeeds $g_{i-1}$.

Let $k$ be the number of long edges contained in $P$, so that 
\[
P = e_1, C_1, \ldots, e_k, C_k
\]
or 
\[
P = e_1, C_1, \ldots, e_k, C_k
\]
and $e_1$ is connected to some interesting cluster of $H^q$.
For $j \in \{1, \ldots, k\}$, let $g_i^j$ be the $i$-th feasible partial embedding created as described above for the $j$-th long edge of $P$.

If $g_i^j$ is that last feasible partial embedding for $e_j$ and $j \neq k$, we can construct $R_i^j$ by copying the embedding of $f_\OPT$ restricted to $\sub_{(H_\OPT, w_\OPT)}(C_j)$ and the path of length $8(c_\OPT)^2+1$ on $\sub_{(H_\OPT, w_\OPT)}(e_{j+1})$ starting from $\sub_{(H_\OPT, w_\OPT)}(C_j)$.
If $g_i^j$ is that last feasible partial embedding for $e_j$, then take $R_i^j$ to be the embedding $f_\OPT$ restricted to the subpath of $\sub_{(H_\OPT, w_\OPT)}(e_j)$ from $A_{g_i^j}$ to $C_j$.

For each $g_i^j$, if $j \neq k$ then let
\[
R_i^j = (g_i^j, j, R_i^j)
\]
and if $j = k$ then let
\[
R_i^j = (g_i^j, j, \emptyset)
\]

By construction, each $R_i^j$ is a feasible partial embedding, and the sequence ordered by $j, i$ forms a sequence of succeeding feasible partial embeddings with the desired attributes.
\end{proof}

\begin{definition}
Let $D(Z^\ell)$ be the directed graph with feasible partial embeddings of $Z$ as vertices, and a directed edge between vertices which succeed one another. We call this graph the \emph{succession graph of $Z$}.
\end{definition}

We present here the $\path$ algorithm.
\begin{description}
\setlength{\itemsep}{-2pt}
\item{\textbf{Step 1.}} Compute $Z^\ell$.
\item{\textbf{Step 2.}} Construct $D(Z^\ell)$. 
\item{\textbf{Step 3.}} Let $F_S = (f_S, 0, \emptyset)$ be the feasible partial embedding of $Z$ implied by $S$. If $F_S \notin V(D(Z^\ell))$ then halt.
\item{\textbf{Step 4.}} If $P = e_1, C_1, e_2, C_2, \ldots, e_j, C_j$:
\begin{description}
\item{\textbf{Step 4.1.}} Perform a DFS of $D(Z^\ell)$, starting from $F_S$. If a node with out-degree 0 is discovered, output the embedding described in Definition \ref{def:path_cluster_embedding}.
\end{description}
\item{\textbf{Step 5.}} If $P = e_1, C_1, e_2, C_2, \ldots, e_j, C_j, e_{j+1}$:
\begin{description}
\item{\textbf{Step 5.1.}} If $T = \emptyset$:
\begin{description}
\item{\textbf{Step 5.1.1.}} Perform a DFS of $D(Z^\ell)$, starting from $F_S$. If a node with out-degree 0 is discovered, output the embedding described in Definition \ref{def:path_cluster_embedding}.
\end{description}
\item{\textbf{Step 5.2.}} If $T \neq \emptyset$:
\begin{description}
\item{\textbf{Step 5.2.1.}} Let $F_T = (f_T, 0, \emptyset)$ be the feasible partial embedding of $Z$ implied by $T$. If $F_T \notin V(D(Z^\ell))$ then halt.
\item{\textbf{Step 5.2.2.}} Perform a DFS of $D(Z^\ell)$, starting from $F_S$. If $F_T$ is discovered, output the embedding described in Definition \ref{def:path_cluster_embedding}.
\end{description}
\end{description}
\end{description}

\begin{lemma} \label{lemma:path_alg_bound}
The $\path$ algorithm runs in time $n^2 \cdot f(H, c_\OPT)$.
\end{lemma}
\begin{proof}
For Step 1, for each vertex $v \in Z$, to decide if $v \in Z^\ell$, explore the neighborhood around $v$ until it is revealed that
\[
|\ball_G(v, 4(c_\OPT)^2_1) \cdot c_\OPT \leq (4(c_\OPT)^2 + 1) \cdot (c_\OPT)^2
\]
or that
\[
|\ball_G(v, 4(c_\OPT)^2_1) \cdot c_\OPT > (4(c_\OPT)^2 + 1) \cdot (c_\OPT)^2.
\]
Therefore, Step 1 can be performed in time $O(n \cdot (c_\OPT)^2)$.

By Lemma \ref{lemma:feasible_limit}, $|V(D(Z^\ell))| = (c_\OPT)^{O(c_\OPT)}$, and so $|E(D(Z^\ell))| = (c_\OPT)^{O(c_\OPT)}$. For $F_1=(f_1, r_1, R_1), F_2=(f_2, r_2, R_2) \in V(D(Z^t))$, there is an edge from $F_1$ to $F_2$ if $F_2$ succeeds $F_1$. We can check if $f_2$ succeeds $f_1$ in time $O(n \cdot 4(c_\OPT)^2)$, and check if $R_1$ and $g_2$ are compatible on $e_{r_2}$ in time $O(n \cdot 4(c_\OPT)^2)$.

For Step 3, we need time $(c_\OPT)^{O(c_\OPT)}$ to find $F_S$ in $V(D(Z^\ell))$. 

For Steps 4.1, 5.1, and 5.2, we use DFS to find an a path in $D(Z^\ell)$, which take time $(c_\OPT)^{O(c_\OPT)}$.

Therefore, the $\path$ algorithm tuns in time $n^2 \cdot f(H, c_\OPT)$.
\end{proof}

\subsection{FPT algorithm}
Given as input $G$, $H$, and an integer $c > 0$, the following algorithm either produces a non-contracting $c$-embedding of $G$ into $(H_\ALG, w_\ALG)$, $H_\ALG$ a subdivision of some quasi-subgraph of $H$, or correctly decides that no such embedding exists.

We provide first an informal summary of the algorithm:
\begin{description}
\setlength{\itemsep}{-1pt}
\item{\textbf{Step 1.}} Choose a quasi-subgraph $H'$ of $H$, and a set of short edges of $H'$.

\begin{description}
\setlength{\itemsep}{-1pt}
\item{\textbf{Step 1.1.}} Find and order the interesting clusters of $H'$. 

\item{\textbf{Step 1.2.}} Find the $\Delta$-interesting vertices of $G$, and choose a subset $I$. Choose an assignment of the vertices in $I$ into the interesting clusters.:

\begin{description}
\setlength{\itemsep}{-1pt}
\item{\textbf{Step 1.2.1.}} For each interesting cluster, use the $\cluster$ algorithm to choose an arrangement of the interesting vertices in the cluster, and along the start of the long edges leaving the cluster.

\begin{description}
\item{\textbf{Step 1.2.1.1.}} Find the path clusters of $H'$ and $\mathcal{T}$, which is the set of connected components  of $G \setminus I$. 

\item{\textbf{Step 1.2.1.2.}} For each connected component in $\mathcal{T}$, choose a path cluster f $H'$ to try embedding it into.

\begin{description}
\item{\textbf{Step 1.2.1.2.1.}} and {\textbf{Step 1.2.1.2.2.}} Using the results of the cluster algorithm above, we know what the embedding into the path cluster should look like near where the path cluster meets an interesting cluster. This determines our inputs to the $\path$ algorithm in the next step.

\item{\textbf{Step 1.2.1.2.3.}} For each path cluster $P_j$, use $\path$ to find an embedding.

\begin{description}
\item{\textbf{Step 1.2.1.2.3.1.}} By construction the embeddings for the interesting and path clusters are compatible on the edges they meet on, so they can be combined into $f_\ALG$ and $(H_\ALG, w_\ALG)$.

\item{\textbf{Step 1.2.1.2.3.2.}} Test $f_\ALG$ and $(H_\ALG, w_\ALG)$ to see if $f_\ALG$ is a non-contracting $c$-embedding of $G$ into $(H_\ALG, w_\ALG)$. If it is, we halt and output $f_\ALG, (H_\ALG, w_\ALG)$. Otherwise, we continue with different choices.

\end{description}

\end{description}

\end{description}

\end{description}

\end{description}

\item{\textbf{Step 2.}} If no embedding is found after all choices are exhausted, output NO.

\end{description}

Here we provide the formal algorithm:

\begin{description}
\item{\textbf{Step 1.}} For each quasi-subgraph $H'$ of $H$ and $S \subseteq E(H')$:

\begin{description}
\item{\textbf{Step 1.1.}} Supposing that $S$ is the set of short edges of $H'$, let $\mathcal{C}^{\geq3}_{H'}$ be the set of interesting clusters of $H'$. 
Let $k = |\mathcal{C}^{\geq3}_{H'}|$.
Fix an ordering $C_1, \ldots, C_k$ of $\mathcal{C}^{\geq3}_{H'}$.

\item{\textbf{Step 1.2.}} Let $I^{\Delta}$ be the set of $\Delta$-interesting vertices of $G$. For each $I \subseteq I^{\Delta}$ and partition $U_{C_1}, \ldots, U_{C_k}$ of $I$:

\begin{description}
\item{\textbf{Step 1.2.1.}} 
For each $i \in \{1, \ldots, k\}$, let $D_i = C_i \cup \{e \in E(H') : e \mbox{ is incident to } C_i\}$, and choose a solution $f_i, (D'_i, w'_i)$ of $\cluster(I, D_i)$ such that for every long edge $e$ incident to $C_i$ we have that $8c^2+2 \leq |\{v \in V(G) : f_i(v) \in \sub_{D'_i}(e)\}|$. Perform the following:

\begin{description}
\item{\textbf{Step 1.2.1.1.}} Let $\mathcal{P}_{H'}$ be the set of path clusters of $H'$.
Let $p = |\mathcal{P}_{H'}|$.
Fix an ordering $P_1, \ldots, P_p$ of $\mathcal{P}$. 
Let $\mathcal{T}$ be the set of connected components of $G \setminus I$. 

\item{\textbf{Step 1.2.1.2.}} For each partition $Q_1, \ldots, Q_p$ of $\mathcal{T}$, and for each $i \in \{1, \ldots, p\}$, let $W_j = \cup_{T \in Q_j} V(T)$:

\begin{description}
\item{\textbf{Step 1.2.1.2.1.}} For each $j \in \{1, \ldots, p\}$, let $x_j \in \{1, \ldots, k\}$ such that $P_j \cap D_{x_j} \neq \emptyset$. 
Let $\{a, b\} \in E(H')$ be a long edge connecting $P_j$ and $D_{x_j}$ with $a \in C_{x_j}$.
Let $S_j = s_1, \ldots, s_{4c^2+1}$ be the sequence of the last $4c^2+1$ consecutive vertices $f_{x_j}$ embeds into $\sub_{D'_{x_j}}(\{a, b\})$, going from $a$ to $b$. 

\item{\textbf{Step 1.2.1.2.2.}} If there exists a second long edge $\{a', b'\} \in E(H')$ connecting $P_j$ to $D_{x'_j}$ from some $x'_j \in \{1, \ldots, k\}$ with $a' \in C_{x'_j}$ then let $T_j = t_1, \ldots, t_{4c^2+1}$ be the sequence of the last $4c^2+1$ consecutive vertices $f_{x_j}$ embeds into $\{a, b\}$, going from $a$ to $b$.
Otherwise, let $T_j = \emptyset$. 

\item{\textbf{Step 1.2.1.2.3.}} For each $j \in \{1, \ldots, p\}$, let $g_j, (P'_j, w'_j)$ be the output of $\path(W_j, P_j, S_j, T_j)$:

\begin{description}
\item{\textbf{Step 1.2.1.2.3.1.}} Construct $f_\ALG, (H_\ALG, w_\ALG)$ from the outputs of the $\cluster$ and $\path$ algorithms as follows:
For all $i \in \{1, \ldots, k\}$ and $j \in \{1, \ldots, p\}$, if $C_i$ and $P_j$ are connected by an edge $e$, then by construction, $f_i$ and $g_j$ are compatible on edge $e$. Combine $f_i$, $g_i$. Call the weighed graph which results from these combinations $(H_\ALG, w_\ALG)$.
Let $f_\ALG$ be the embedding of $G$ into $(H_\ALG, w_\ALG)$.

\item{\textbf{Step 1.2.1.2.3.2.}} If $f_\ALG$ is a non-contracting, $c$-embedding of $G$ into $(H_\ALG, w_\ALG)$, then output $f_\ALG, (H_\ALG, w_\ALG)$ and halt.

\end{description}

\end{description}

\end{description}

\end{description}

\end{description}

\item{\textbf{Step 2.}} Output NO.

\end{description}


\begin{lemma} \label{lemma:important_bound}
If $v \in V(G)$ is $\Delta$-interesting, then there exists $z \in V(H^q)$ such that 
\[
d_{(H_\OPT, w_\OPT)}(f_\OPT(v), z) \leq 2 c_\OPT \cdot \Delta.
\]
\end{lemma}
\begin{proof}
Let $v \in V(G)$ be a $\Delta$-interesting vertex. Suppose that for all $p \in V(H^q)$ we have that
\[
d_{(H_\OPT, w_\OPT)}(f_\OPT(v), p) > 2 c_\OPT \cdot \Delta,
\]
and we shall find a contradiction.

Since $f_\OPT$ is a non-contracting $c_\OPT$-embedding, for all $x \in \ball_G(v, \Delta)$, we have
\[
d_{(H_\OPT, w_\OPT)}(f_\OPT(v), f_\OPT(x)) \leq c_\OPT \cdot \Delta,
\]
and for all $x, y \in \ball_G(v, \Delta)$,
\begin{align*}
d_{(H_\OPT, w_\OPT)}(f_\OPT(x), f_\OPT(y)) 
&\leq d_{(H_\OPT, w_\OPT)}(f_\OPT(v), f_\OPT(x)) + d_{(H_\OPT, w_\OPT)}(f_\OPT(v), f_\OPT(y)) \\
&\leq 2 c_\OPT \cdot \Delta.
\end{align*}
Since for all $p \in V(H^q)$ we have that $d_{(H_\OPT, w_\OPT)}(f_\OPT(v), p) > 2 c_\OPT \cdot \Delta$, we have that $f_\OPT$ embeds all $x \in \ball_G(v, \Delta)$ into $\{a, b\}_\OPT$, for some $\{a, b\} \in E(H^q)$. 
From the limits stated above, we have that 
\[
\min_{x \in \ball_G(v, \Delta)}\{d_{(H_\OPT, w_\OPT)}(a, f_\OPT(x))\} > c_\OPT \cdot \Delta
\]
and 
\[
\min_{x \in \ball_G(v, \Delta)}\{d_{(H_\OPT, w_\OPT)}(b, f_\OPT(x))\} > c_\OPT \cdot \Delta.
\]
Therefore, for any $x, y \in \ball_G(v, \Delta)$, the shortest path from $f_\OPT(x)$ to $f_\OPT(y)$ in $(H_\OPT, w_\OPT)$ is the path from $f_\OPT(x)$ to $f_\OPT(y)$ contained in $\{a, b\}_\OPT$.
Let 
\[
z = \argmin_{x \in \ball_G(v, \Delta)}\{d_{(H_\OPT, w_\OPT)}(a, f_\OPT(x))\},
\]
and let $g$ be a function so that for any $x \in \ball_G(v, \Delta)$,
\[
g(x) = d_{(H_\OPT, w_\OPT)}(f_\OPT(z), f_\OPT(x)).
\]
Then for all $s, t \in V(G)$, we have that
\begin{align*}
|g(s) - g(t)| &= |d_{(H_\OPT, w_\OPT)}(f_\OPT(s), f_\OPT(z)) - d_{(H_\OPT, w_\OPT)}(f_\OPT(t), f_\OPT(z))| \\
&= d_{(H_\OPT, w_\OPT)}(f_\OPT(s), f_\OPT(t))
\end{align*}
and therefore
\[
d_G(s, t) \leq |g(s) - g(t)| \leq c_\OPT \cdot d_G(s, t).
\]
Therefore, $g$ is a non-contracting, $c_\OPT$-embedding of $\ball_G(v, \Delta)$ into the line, which is a contradiction.
Thus the supposition, that for all $p \in V(H^q)$
\[
d_{(H_\OPT, w_\OPT)}(f_\OPT(v), p) > 2 c_\OPT \cdot \Delta,
\]
is false.
\end{proof}

\begin{lemma} \label{lemma:Delta_bound}
Let $I^\Delta$ be the set of $\Delta$-interesting vertices of $G$.
Then
\[
|I^\Delta| \leq 8 c_\OPT \cdot \Delta \cdot |E(H)|. 
\]
\end{lemma}
\begin{proof}
By Lemma \ref{lemma:important_bound}, each $\Delta$-interesting vertex is within distance $2 c_\OPT \cdot \Delta$ of a vertex of $V(H^q)$.
Since $f_\OPT$ is non-contracting, for each $e \in E(H^q)$, $f_\OPT$ can map at most
\[
2 c_\OPT \cdot \Delta + 2 c_\OPT \cdot \Delta
\]
$\Delta$-interesting vertices to $e$.
Therefore, there are at most
\[
4 c_\OPT \cdot \Delta \cdot |E(H^q)|
\]
$\Delta$-interesting vertices in $G$.
Thus, 
\[
|I^\Delta| \leq 8 c_\OPT \cdot \Delta \cdot |E(H)|.
\]
\end{proof}

\begin{lemma} \label{lemma:H_important}
Let $C$ be any interesting cluster of $H^q$. For any $v \in V(G)$ such that 
\[
d_{(H_\OPT, w_\OPT)}(f_\OPT(v), f_\OPT(C)) \leq 8c_\OPT,
\]
 we have that $v$ is $\diam(C)$-interesting.
\end{lemma}
\begin{proof}
Since $C$ is an interesting cluster of $H^q$, there exist long edges $e_1, e_2, e_3 \in E(H^q)$ adjacent to $C$. Since $f_\OPT$ is proper, there exists vertices $a_1, a_2, a_3, b_1, b_2, b_3 \in V(G)$ such that for all $i \in \{1, 2, 3\}$,
\[
f_\OPT(a_i) \in {e_i}_\OPT,
\]
\[
f_\OPT(b_i) \in {e_i}_\OPT,
\]
\[
4(c_\OPT)^3 + 2c_\OPT \leq d_{(H_\OPT, w_\OPT)}(f_\OPT(a_i), V(C)) \leq 4(c_\OPT)^3 + 4c_\OPT,
\]
and 
\[
0 \leq d_{(H^q, w^q)}(f_\OPT(b_i), V(C)) \leq 2c_\OPT.
\]
Since $f_\OPT$ is a $c_\OPT$-embedding, for each $i \in \{1, 2, 3\}$ we have that
\[
4(c_\OPT)^3 \geq d_G(a_i, b_i) \geq 4(c_\OPT)^2.
\]

Let 
\[
V_C = \{v \in V(G) : d_{(H^q, w^q)}(f_\OPT(v), (C, w^q)) \leq 2(c_\OPT)^2\}.
\]
Suppose that $b_1$ and $b_2$ are in distinct connected components $C_1$, $C_2$ of $G[V_C]$. Let $P_{1,2}$ be the shortest path in $\ball_{(H_\OPT, w_\OPT)}(C, (c_\OPT)^2)$ from $f_\OPT(b_1)$ to $f_\OPT(b_2)$.
Let $p_1$ be the vertex in $C_1$ such that $f_\OPT(p_1) \in P_{1,2}$ and $d_{(H_\OPT, w_\OPT)}(f_\OPT(b_1), f_\OPT(p_1))$ is maximal. Let $p \in V(G) \setminus V(C_1)$ such that $f_\OPT(p)$ is in the subpath of $P_{1,2}$ from $f_\OPT(p_1)$ to $f_\OPT(b_2)$ and $d_{(H_\OPT, w_\OPT)}(f_\OPT(p_1), f_\OPT(p))$ is minimal.
Since $f_\OPT$ is proper we have
\[
d_{(H_\OPT, w_\OPT)}(f_\OPT(p_1), f_\OPT(p)) \leq 2c_\OPT,
\]
and since $f_\OPT$ is non-contracting we have
\[
d_G(p_1, p) \leq 2c_\OPT.
\]
Let $S \subseteq V(G)$ be the set of vertices in the shortest path from $p_1$ to $p$ in $G$.
For all $s \in S$,
\[
f_\OPT(s) \in \ball_{(H_\OPT, w_\OPT)}(f_\OPT(p_1), 2c_\OPT \cdot c_\OPT),
\]
and since $f_\OPT(p_1) \in P_{1,2}$, we have
\[
f_\OPT(s) \in \ball_{(H_\OPT, w_\OPT)}(C, 2(c_\OPT)^2).
\]
Therefore, $p_1$ and $p$ are in the same connected component of $G[V_C]$, and thus $p$ and $b_1$ are in the same connected component of $G[V_C]$. Therefore $p \in C_1$ and $p \in V(G) \setminus V(C_1)$, a contradiction.
Therefore, $b_1$ and $b_2$ are in the same connected component of $G[V_C]$, and by a similar argument, $b_2$ and $b_3$ are in the same connected component of $G[V_C]$. Thus, $b_1, b_2, b_3$ are all in the same connected component of $G[V_C]$.

For all $i \in \{1, 2, 3\}$, let $P_i$ be the set containing the vertices in the shortest path in $G$ from $a_i$ to $b_i$, and let $P^{c_\OPT}_i$ be the set containing the first $c_\OPT$ vertices in the shortest path in $G$ from $a_i$ to $b_i$.
Suppose there exists $i, j \in \{1, 2, 3\}$, $i \neq j$, such that
\[
P^{c_\OPT}_i \cap P^{c_\OPT}_j \neq \emptyset.
\]
Then there exists $y \in P^{c_\OPT}_i \cap P^{c_\OPT}_j$, and 
\[
d_G(y, a_i) \leq c_\OPT
\]
and 
\[
d_G(y, a_j) \leq c_\OPT.
\]
Therefore, since $f_\OPT$ is a $c_\OPT$-embedding, we have that
\[
d_{(H_\OPT, w_\OPT)}(f_\OPT(y), f_\OPT(a_i)) \leq c_\OPT^2
\]
and
\[
d_{(H_\OPT, w_\OPT)}(f_\OPT(y), f_\OPT(a_j)) \leq c_\OPT^2.
\]
Since $f_\OPT(a_i), f_\OPT(a_j)$ are not in the same edge of $(H_\OPT, w_\OPT)$, and each are of distance greater than $(c_\OPT)^2$ from either of the endpoints of the edges containing $f_\OPT(a_i), f_\OPT(a_j)$, this is a contradiction. Therefore, for all $i, j \in \{1, 2, 3\}$, $i \neq j$, we have that
\[
P^{c_\OPT}_i \cap P^{c_\OPT}_j = \emptyset.
\]
Furthermore, since for all $i \in \{1,2,3\}$ we have that
\[
4(c_\OPT)^3 + 2c_\OPT \leq d_{(H_\OPT, w_\OPT)}(f_\OPT(a_i), (C, ^q)) \leq 4(c_\OPT)^3 + 4c_\OPT,
\]
for all $p \in P^{c_\OPT}_i$, we have that
\begin{align*}
d_{(H_\OPT, w_\OPT)}(f_\OPT(p), V(C)) &\geq 4(c_\OPT)^3 + 2c_\OPT - (c_\OPT)^2 \\
&\geq 3(c_\OPT)^3 + 2c_\OPT,
\end{align*}
and therefore $P^{c_\OPT}_i \cap V_C = \emptyset$.

Let $C_{1,2,3}$ be the connected component of $G[V_C \cup P_1 \cup P_2 \cup P_3]$ containing $b_1, b_2, b_3$. For all $i \in \{1, 2, 3\}$, $C_{1,2,3}$ contains a path connecting $a_i$ and $b_i$, with at least $c_\OPT$ vertices not in $V_C$. Therefore, $C_{1,2,3}$ consists of a central component with at least 3 paths of length $\geq c_\OPT$ leaving the central component. Such a structure cannot be embedding into the line with distortion $c_\OPT$.

Let $v \in \ball_{(H_\OPT, w_\OPT)}(V(C), \diam(C) + 2(c_\OPT)^2)$. Then for all $z \in V(C_{1,2,3})$, we have
\begin{align*}
d_G(v, z) &\leq d_{(H^q, w^q)}(f_\OPT(v), f_\OPT(z)) \\
&\leq \max_{i\in\{1,2,3\}} d_{(H^q, w^q)}(f_\OPT(v), f_\OPT(a_i)) \\
&\leq \diam(C) + 8(c_\OPT)^3.
\end{align*}
Therefore, for all $v \in \ball_{(H_\OPT, w_\OPT)}(V(C), \diam(C) + 2(c_\OPT)^2)$, we have that
\[
\ball_G(v, \diam(C) + 8(c_\OPT)^3) \subseteq \ball_G(v, \Delta)
\]
does not embed into the line.
Therefore, $v$ is $\Delta$-interesting.
\end{proof}

\begin{lemma} \label{lemma:path_cluster_vertices}
Let 
\[
I_\OPT = \{v \in V(G) : \exists C \in \mathcal{C}^{\geq 3} \mbox{ such that } v \in \ball_{(H_\OPT, w_\OPT)}(C, 2(c_\OPT)^2)\},
\]
and let $C_G$ be any connected component of $G \setminus I_\OPT$. Then there exists a path cluster $P$ of $(H_\OPT, w_\OPT)$ such that $f_\OPT(V(C_G)) \subseteq (P, w_\OPT)$.
\end{lemma}
\begin{proof}
Suppose there exists $C$, a connected component of $G \setminus I_\OPT$ such that for some $x, y \in V(C)$, we have that $f_\OPT(x)$ and $f_\OPT(y)$ are in different path clusters of $H^q$. Therefore, any path in $(H_\OPT, w_\OPT)$ between $f_\OPT(x)$ and $f_\OPT(y)$ must intersect the subdivision of an interesting cluster of $H^q$.
So for any path in $G$ between  $x$ and $y$, the path must contain a vertex $z$ such that, for $C'$ the subdivision of some interesting cluster of $H^q$, 
\[
d_{(H_\OPT, w_\OPT)}(f_\OPT(z), V(C')) \leq c_\OPT,
\]
and thus $x$ and $y$ cannot be in the same connected component of $G \setminus I_\OPT$.
Therefore, for all connected components $C$ of $G \setminus I_\OPT$, there exists a path cluster $P$ of $H^q$ such that $f_\OPT(C) \subseteq P$.
\end{proof}

\begin{lemma} \label{lemma:path_cluster_bound}
Let 
\[
I_\OPT = \{v \in V(G) : \exists C \in \mathcal{C}^{\geq 3} \mbox{ such that } v \in \ball_{(H_\OPT, w_\OPT)}(C, 2(c_\OPT)^2)\},
\]
and let $T$ be the set of connected components of $G \setminus I_\OPT$.
Then $|T| \leq (4c_\OPT \cdot |E(H)|)^2$.
\end{lemma}
\begin{proof}
Let $C_G$ be a connected component of $G \setminus I_\OPT$.
Since $G$ is a connected graph, we have that
\[
\ball_G(V(C_G), 1) \cap I_\OPT \neq \emptyset.
\]
Let $z \in \ball_G(V(C_G), 1) \cap I_\OPT$. Since $f_\OPT$ is a $c_\OPT$-embedding, we have that
\[
z \in I_\OPT \setminus \{v \in V(G) : \exists C \in \mathcal{C}^{\geq 3} \mbox{ such that } v \in \ball_{(H_\OPT, w_\OPT)}(C, 2(c_\OPT)^2 - c_\OPT)\}.
\]
For each edge in $E(H^q)$, there are at most $2c_\OPT$ vertices in $v \in V(G)$ such that 
\[
v \in I_\OPT \setminus \{v \in V(G) : \exists C \in \mathcal{C}^{\geq 3} \mbox{ such that } v \in \ball_{(H_\OPT, w_\OPT)}(V(C), 2(c_\OPT)^2 - c_\OPT)\},
\]
and so there are at most $2c_\OPT \cdot |E(H^q)| \leq 4c_\OPT \cdot |E(H)|$ vertices to which each connected component of $G \setminus I_\OPT$ is connected to one or more. From Lemma \ref{lemma:local_density}, for all $v \in V(G)$, we have that
\[
|\ball_{G}(v, 1)| \leq 4c_\OPT \cdot |E(H)|.
\]
Therefore, there are at most $(4c_\OPT \cdot |E(H)|)^2$ connected components of $G \setminus I_\OPT$.
\end{proof}

\begin{lemma} \label{lemma:fpt_H_runtime}
The FPT Algorithm runs in time $n^{\cO(1)} \cdot f(H, c_\OPT)$.
\end{lemma}
\begin{proof}
For Step 1, when creating a quasi-subgraph, the only rule which increases the number of edges is rule 3. Since The quasi-subgraph must be connected, rule 3 can be applied at most once for each edge in $E(H)$. Therefore, 
\[
|E(H')| \leq 2|E(H)|
\] 
and 
\[
|V(H')| \leq 2|V(H)|.
\]
We can find an upper bound on the number of quasi-subgraphs of $H$ by first selecting a subset of edges of $H$ to apply rule 3 to, in which we have $2^{|E(H)|}$ choices, and then selecting subsets of vertices and edges for deletion, of which there are at most $2^{2|V(H)|}$ and $2^{2|E(H)|}$ sets to choose from. There are therefore at most
\[
2^{|E(H)|} \cdot 2^{2|V(H)|} \cdot 2^{2|E(H)|} = 2^{|V(H)|} \cdot 2^{3|E(H)|}
\]
choices for quasi-subgraph of $H$, and therefore
\[
2^{3|E(H)|} \cdot 2^{2|E(H)|} = 2^{5|E(H)|}
\]
possible choices for Step 1.

For Step 1.1, we can find the interesting clusters of $H'$ in time $f(H)$.

For Step 1.2, Lemma \ref{lemma:Delta_bound} tells us that 
\[
|I^\Delta| \leq 8 c_\OPT \cdot \Delta \cdot |E(H)|. 
\]
Therefore, there are at most
\[
2^{8 c_\OPT \cdot \Delta \cdot |E(H)|}
\]
subsets of $I^\Delta$.
Each interesting cluster of $H'$ must contain a vertex of $H'$, and so there are at most 
\[
2|V(H)|^{|I|} \leq 2|V(H)|^{8 c_\OPT \cdot \Delta \cdot |E(H)|}
\]
possible partitions of $I$.
Therefore, there are at most 
\[
2^{8 c_\OPT \cdot \Delta \cdot |E(H)|} \cdot 2|V(H)|^{8 c_\OPT \cdot \Delta \cdot |E(H)|}
\]
choices for Step 1.2.

For Step 1.2.1, we have that 
\[
|I| \leq |I^\Delta| \leq 8 c_\OPT \cdot \Delta \cdot |E(H)|,
\]
and each $D_i$ is a subgraph of $H'$, and thus by Lemma \ref{lemma:cluster_alg_bound}, each instance of $\cluster(I, D_i)$ has at most
\[
O(2|E(H)|^{|I|} \cdot |I|! \cdot (2|V(H)| - 2)!) = |E(H)|^{O(|E(H)|}
\]
solutions.

For Step 1.2.1.1, we can find the connected components of $G \setminus I$ in time $O(n)$.

For Step 1.2.1.2, by Lemma \ref{lemma:path_cluster_bound}, we have that
\[
|\mathcal{T}| \leq (4c_\OPT \cdot |E(H)|)^2.
\]
Each path cluster contains at least one long edge of $H'$, so there are at most $2|E(H)|$ path clusters.
Therefore, there are at most
\[
(2|E(H)|)^{|\mathcal{T}|} = (2|E(H)|)^{(4c_\OPT \cdot |E(H)|)^2}
\]
possible partitionings of $\mathcal{T}$.

Steps 1.2.1.2.1 and 1.2.1.2.2 can be done in time $f(H, c)$.

For Step 1.2.1.2.3, by Lemma \ref{lemma:path_alg_bound}, the $\path$ algorithm runs in time $n^2 \cdot f(H, c_\OPT)$.

Step 1.2.1.2.3.1 can be done in time $\cO(n)$, by checking where each vertex in $G$ is embedded.

Step 1.2.1.2.3.2 can be done by computing all-pairs shortest path on both $G$ and $(H_\ALG, w_\ALG)$, then for each $u, v \in V(G)$, compare $d_G(u, v)$ and $d_{(H_\ALG, w_\ALG)}(f_\ALG(u), f_\ALG(v))$. Since each edge of $H^q$ is subdivided no more than $c_\OPT\cdot  n$ times, $|V(H_\ALG)| \leq 2|V(H)|c_\OPT \cdot n$, and so this check can be performed in time $\cO(f(H, c_\OPT) \cdot n^3)$.

Therefore, the algorithm runs in time $n^{\cO(1)} \cdot f(H, c_\OPT)$.
\end{proof}

\begin{lemma} \label{lemma:fpt_H_correctness}
If $H$ contains an interesting cluster and $c \geq c_\OPT$, then the FPT Algorithm outputs $f_\ALG, (H_\ALG, w_\ALG)$, where $f_\ALG$ is a non-contracting $c$-embedding of $G$ into $(H_\ALG, w_\ALG)$.
\end{lemma}
\begin{proof}
Since the algorithm iterates over all possible choices of quasi-subgraphs and short edges, we may assume that for some iteration, $H' = H^q$, and the correct short edges are chosen. By Lemma \ref{lemma:importance_tractable}, we can find all $\Delta$-interesting vertices. By Lemma \ref{lemma:H_important} and Definition \ref{def:Delta}, we have that all vertices which $f_\OPT$ embeds into a radius of $8\cdot(c_\OPT)^4$ of any interesting cluster is $\Delta$-interesting. Therefore, since the algorithm tries all assignments of $\Delta$-interesting vertices to interesting clusters, and all possible orders in which the vertices might be embedded along the edges of and incident to the interesting clusters, we may assume that the algorithm will reach a state where for each interesting cluster $C$, $f_\OPT$ and the algorithm match for each edge $e \in E(C)$ on the vertices embedded into $e$, the order of the vertices on $e$, and the order of vertices embedded into long edges leaving $C$, up to distance at least $8\cdot(c_\OPT)^2+2$.

For each path cluster, for each long edge in the path cluster connected to an interesting cluster, the $\path$ algorithm is given as input a sequence of $4(c_\OPT)^2+1)$ vertices of distance at least $4(c_\OPT)^2+1)$ from the interesting cluster, and in the order they are embedded, when traversing the edge away from the interesting cluster.
By Lemma \ref{lemma:feasible_partial_embeddings_fopt}, for each path cluster $P$, there exists a solution to the $\path$ algorithm such that if $P$ is connected by long edge $e=\{a, b\}$ to interesting cluster $C$, and $a \in V(C)$, then the solution is compatible with $f_\OPT$ restricted $\ball_{(H_\OPT, w_\OPT)}(V(C), 8(c_\OPT)^2)$.

Therefore, we may assume that the algorithm computes $f_\ALG$, $(H_\ALG, w_\ALG)$ such that
\begin{enumerate}
\item $H_\ALG$ is a subdivision of $H^q$.
\item For each interesting cluster $C_I$ of $H^q$, for each $e \in E(C_I)$,
\[
f_\OPT(V(G)) \cap e_\OPT = f_\ALG(V(G)) \cap \sub_{(H_\ALG, w_\ALG)}(e)
\]
and the order from imposed on $f_\OPT(V(G)) \cap e_\OPT$ by $f_\OPT$ is the same as the order imposed on $f_\ALG(V(G)) \cap \sub_{(H_\ALG, w_\ALG)}(e)$ by $f_\ALG$.
\item For each path cluster $C_P$ in $H^q$, we have that
\[
f_\OPT(V(G)) \cap \sub_{(H_\OPT, w_\OPT)}(C_P) = f_\ALG(V(G)) \cap \sub_{(H_\ALG, w_\ALG)}(C_P)
\]
\item For any path cluster $C_P$ in $H^q$, for any $\{a, b\} \in E(G)$ such that $f_\OPT(a) \in \sub_{(H_\OPT, w_\OPT)}(C_P)$ and $f_\OPT(b) \in \sub_{(H_\OPT, w_\OPT)}(C_P)$, we have that
\[
1 \leq d_{(H_\ALG, w_\ALG)}(f_\ALG(a), f_\ALG(b)) \leq c_\OPT.
\]
\item For any path cluster $C_P$ in $H^q$, for any $u, v \in V(G)$ such that $f_\OPT(u) \in \sub_{(H_\OPT, w_\OPT)}(C_P)$ and $f_\OPT(v) \in \sub_{(H_\OPT, w_\OPT)}(C_P)$, we have that
\[
d_{(H_\ALG, w_\ALG)}(f_\ALG(u), f_\ALG(v)) \geq d_{(H_\OPT, w_\OPT)}(f_\OPT(u), f_\OPT(v)).
\]
\end{enumerate}

Let $u, v \in V(G)$, and let $P_{u, v}$ be the shortest path in $G$ from $u$ to $v$.

If there exists an interesting cluster $C_I$ in $H^q$ such that $f_\OPT(P_{u, v}) \in \sub_{(H_\OPT, w_\OPT)}(C_I)$, then the $\cluster$ algorithm has ensured that
\[
d_G(u, v) \leq d_{(H_\ALG, w_\ALG)}(f_\ALG(u), f_\ALG(v)) \leq c_\OPT \cdot d_G(u, v).
\]

If there exists a  path cluster $C_P$ in $H^q$ such that $f_\OPT(P_{u, v}) \in \sub_{(H_\OPT, w_\OPT)}(C_P)$, then by the observations above, we have that
\begin{align*}
d_G(u, v)
&\leq d_{(H_\OPT, w_\OPT)}(f_\OPT(u), f_\OPT(v)) \\
&\leq d_{(H_\ALG, w_\ALG)}(f_\ALG(u), f_\ALG(v)) \\
&\leq \sum_{e \in P_{u, v}} c_\OPT \\
&\leq c_\OPT \cdot d_G(u, v).
\end{align*}

If $u$ and $v$ are not in the same interesting or path cluster, then there is some minimum sequence $C_1, C_2, \ldots, C_k$ such that $f_\OPT(P_{u, v}) \in \sub_{(H_\OPT, w_\OPT)}(C_1 \cup C_2 \cup \ldots \cup C_k)$. Since the embeddings on these clusters are compatible, for each $i \in \{1, \ldots, k-1\}$, there is a sequence of $4(c_\OPT)^2+1$ consecutive vertices embedded in the edge connecting $C_i$ and $C_{i+1}$. For each $i \in \{1, \ldots, k-1\}$, there exists $v_i \in V(P_{u, v})$ such that $v_i$ intersects the sequence between $C_i$ and $C_{i+1}$. Therefore,
\begin{align*}
d_{(H_\ALG, w_\ALG)}(f_\ALG(u), f_\ALG(v))
&= d_{(H_\ALG, w_\ALG)}(f_\ALG(u), f_\ALG(v_1)) + \ldots + d_{(H_\ALG, w_\ALG)}(f_\ALG(v_{k-1}), f_\ALG(v)) \\
&\leq c_\OPT \cdot d_G(u, v_1) + \ldots + c_\OPT \cdot d_G(v_{k-1}, v) \\
&= c_\OPT \cdot d_G(u, v).
\end{align*}
Since the $\cluster$ and $\path$ algorithms do not allow contraction of distances, we also have that
\begin{align*}
d_{(H_\ALG, w_\ALG)}(f_\ALG(u), f_\ALG(v))
&= d_{(H_\ALG, w_\ALG)}(f_\ALG(u), f_\ALG(v_1)) + \ldots + d_{(H_\ALG, w_\ALG)}(f_\ALG(v_{k-1}), f_\ALG(v)) \\
&\geq d_G(u, v_1) + \ldots + d_G(v_{k-1}, v) \\
&= d_G(u, v).
\end{align*}

Therefore, $f_\ALG$ is a non-contracting, $c_\OPT$-embedding of $G$ into $(H_\ALG, w_\ALG)$, where $H_\ALG$ is some subdivision of $H^q$, and $H^q$ is a quasi-subgraph of $H$.
\end{proof}

We are now ready to prove Theorem \ref{theorem:fpt_H}.

\begin{proof}[Proof of Theorem \ref{theorem:fpt_H}]
To ensure the existence of interesting vertices in $G$, and thus interesting clusters in $H$, we make the following modifications to $G$ and $H$.

Let $k = 8 c \cdot |E(H)|$, and $K_k$ the complete graph on $k$ vertices. Note that there cannot be a non-contracting $c$-embedding of $K_k$ into $H$, since for at least one edge of any quasi-subgraph $H^q$ of $H$ and $|E(H^q)| \leq 2|E(H)|$, by the pigeonhole principle, for any embedding, at least $4c$ vertices of $G$ would be embedded into the same subdivsion of an edge of $H^q$, and so two vertices adjacent in $K_k$ would be embedded with distance greater than $c$.

We will now describe how to use $K_k$ to find a non-contracting $c$ embedding of $G$ into $H$, if such an embedding exists.

Create a new graph $G'$ in the following way: Connect a single copy of $K_k$ to $G$ by creating 3 paths of length $16c^4+1$ from a single vertex of $K_k$ to 3 arbitrary vertices $v_1, v_2, v_3$ in $G$. Let $\mathbb{H}_k$ be the set of graphs creating in the following way: For all $A \subset V(H) \cup E(H)$ with $|A|=3$, connect single copy of $K_k$ to $H$ by connecting $K_k$ to each $a \in A$ through a single vertex of $K_k$. If $a \in V(H)$ then this connection is by an edge from $K_k$ to $a$, and if $a \in E(H)$ then the connection is by subdividing $a$ and connecting $K_k$ to the new vertex.

If there exists a non-contracting $c$-embedding of $G$ into $H$, then there must exist $H' \in \mathbb{H}_k$ such that there is a non-contracting $c$-embedding of $G'$ into $H'$. This is easy to see by taking the embedding of $G$ into $H$ and extending it. Construct $H''$ by connecting a vertex of $K_k$ to the subdivision of $H$ (call this subdivision $H_s$) used for the embedding with 3 paths of length $16c^4+1$ to the vertices $v_1, v_2, v_3$ are embedded to. Using the embedding of $G$ into $H$, it is clear that there exists $H' \in \mathbb{H}_k$ such that a subdivision of $H'$ matches $H''$. By how $H''$ was constructed, the additional vertices in $V(G') \setminus V(G)$ are all embedded into vertices in $V(H'') \setminus V(H_s)$. By modifying our algorithm so that only embeddings of this type are considered, if a non-contracting $c$-embedding of $G$ into $H$ exists, the corresponding non-contracting $c$-embedding of $G'$ into $H'$ can be found, and then the corresponding non-contracting $c$-embedding of $G$ into $H$ can be extracted.

The rest of the theorem follows immediately from Lemma \ref{lemma:fpt_H_correctness} and Lemma \ref{lemma:fpt_H_runtime}.
\end{proof}

\bibliographystyle{plain}
\bibliography{embedding}

\begin{thebibliography}{10}

\bibitem{arora2008euclidean}
Sanjeev Arora, James Lee, and Assaf Naor.
\newblock Euclidean distortion and the sparsest cut.
\newblock {\em Journal of the American Mathematical Society}, 21(1):1--21,
  2008.

\bibitem{arora2009expander}
Sanjeev Arora, Satish Rao, and Umesh Vazirani.
\newblock Expander flows, geometric embeddings and graph partitioning.
\newblock {\em Journal of the ACM (JACM)}, 56(2):5, 2009.

\bibitem{BadoiuCIS05}
Mihai B{\u{a}}doiu, Julia Chuzhoy, Piotr Indyk, and Anastasios Sidiropoulos.
\newblock Low-distortion embeddings of general metrics into the line.
\newblock In {\em Proceedings of the 37th Annual ACM Symposium on Theory of
  Computing (STOC)}, pages 225--233. ACM, 2005.

\bibitem{badoiu2006embedding}
Mihai B{\u{a}}doiu, Julia Chuzhoy, Piotr Indyk, and Anastasios Sidiropoulos.
\newblock Embedding ultrametrics into low-dimensional spaces.
\newblock In {\em Proceedings of the 22ndannual symposium on Computational
  geometry (SoCG)}, pages 187--196. ACM, 2006.

\bibitem{BadoiuDGRRRS05}
Mihai B{\u{a}}doiu, Kedar Dhamdhere, Anupam Gupta, Yuri Rabinovich, Harald
  R{\"a}cke, R.~Ravi, and Anastasios Sidiropoulos.
\newblock Approximation algorithms for low-distortion embeddings into
  low-dimensional spaces.
\newblock In {\em Proceedings of the 16th Annual ACM-SIAM Symposium on Discrete
  Algorithms (SODA)}, pages 119--128. SIAM, 2005.

\bibitem{BadoiuIS07}
Mihai B{\u{a}}doiu, Piotr Indyk, and Anastasios Sidiropoulos.
\newblock Approximation algorithms for embedding general metrics into trees.
\newblock In {\em Proceedings of the 18th Annual ACM-SIAM Symposium on Discrete
  Algorithms (SODA)}, pages 512--521. ACM and SIAM, 2007.

\bibitem{bartal1996probabilistic}
Yair Bartal.
\newblock Probabilistic approximation of metric spaces and its algorithmic
  applications.
\newblock In {\em Foundations of Computer Science, 1996. Proceedings., 37th
  Annual Symposium on}, pages 184--193. IEEE, 1996.

\bibitem{ChandranMOPSS08}
Nishanth Chandran, Ryan Moriarty, Rafail Ostrovsky, Omkant Pandey, Mohammad~Ali
  Safari, and Amit Sahai.
\newblock Improved algorithms for optimal embeddings.
\newblock {\em {ACM} Transactions on Algorithms}, 4(4), 2008.

\bibitem{onak2013fat}
Mark de~Berg, Krzysztof Onak, and Anastasios Sidiropoulos.
\newblock Fat polygonal partitions with applications to visualization and
  embeddings.
\newblock {\em Journal of Computational Geometry}, 4(1):212–239, 2013.

\bibitem{edmonds2010inapproximability}
Jeff Edmonds, Anastasios Sidiropoulos, and Anastasios Zouzias.
\newblock Inapproximability for planar embedding problems.
\newblock In {\em Proceedings of the 20th Annual ACM-SIAM Symposium on Discrete
  Algorithms (SODA)}, pages 222--235. Society for Industrial and Applied
  Mathematics, 2010.

\bibitem{farach1995robust}
Martin Farach, Sampath Kannan, and Tandy Warnow.
\newblock A robust model for finding optimal evolutionary trees.
\newblock {\em Algorithmica}, 13(1-2):155--179, 1995.

\bibitem{farach1999approximate}
Martin Farach-Colton and Piotr Indyk.
\newblock Approximate nearest neighbor algorithms for hausdorff metrics via
  embeddings.
\newblock In {\em Proceedings of the 40th Annual Symposium on Foundations of
  Computer Science (FOCS)}, pages 171--179. IEEE, 1999.

\bibitem{Fellows:2013:DFP:2539126.2489789}
Michael Fellows, Fedor~V. Fomin, Daniel Lokshtanov, Elena Losievskaja, Frances
  Rosamond, and Saket Saurabh.
\newblock Distortion is fixed parameter tractable.
\newblock {\em ACM Trans. Comput. Theory}, 5(4):16:1--16:20, November 2013.

\bibitem{FellowsFLLSR09}
Michael~R. Fellows, Fedor~V. Fomin, Daniel Lokshtanov, Elena Losievskaja,
  Frances~A. Rosamond, and Saket Saurabh.
\newblock Distortion is fixed parameter tractable.
\newblock In {\em Proceedings of the 36th International Colloquium on Automata,
  Languages and Programming (ICALP)}, volume 5555 of {\em Lecture Notes in
  Computer Science}, pages 463--474. Springer, 2009.

\bibitem{HallP05}
Alexander Hall and Christos~H. Papadimitriou.
\newblock Approximating the distortion.
\newblock In {\em Approximation, Randomization and Combinatorial Optimization,
  Algorithms and Techniques, 8th International Workshop on Approximation
  Algorithms for Combinatorial Optimization Problems (APPROX-RANDOM)}, volume
  3624 of {\em Lecture Notes in Computer Science}, pages 111--122. Springer,
  2005.

\bibitem{Indyk01}
Piotr Indyk.
\newblock Algorithmic applications of low-distortion geometric embeddings.
\newblock In {\em Proceedings of the 42nd IEEE Symposium on Foundations of
  Computer Science (FOCS)}, pages 10--33. IEEE, 2001.

\bibitem{indyk2006stable}
Piotr Indyk.
\newblock Stable distributions, pseudorandom generators, embeddings, and data
  stream computation.
\newblock {\em Journal of the ACM (JACM)}, 53(3):307--323, 2006.

\bibitem{indyk2004low}
Piotr Indyk and Jiri Matousek.
\newblock Low-distortion embeddings of finite metric spaces.
\newblock In {\em in Handbook of Discrete and Computational Geometry}, pages
  177--196. CRC Press, 2004.

\bibitem{KenyonRS04}
Claire Kenyon, Yuval Rabani, and Alistair Sinclair.
\newblock Low distortion maps between point sets.
\newblock In {\em Proceedings of the 36th Annual ACM Symposium on Theory of
  Computing (STOC)}, pages 272--280. ACM, 2004.

\bibitem{KenyonRS09}
Claire Kenyon, Yuval Rabani, and Alistair Sinclair.
\newblock Low distortion maps between point sets.
\newblock {\em {SIAM} J. Comput.}, 39(4):1617--1636, 2009.

\bibitem{khot2007hardness}
Subhash Khot and Rishi Saket.
\newblock Hardness of embedding metric spaces of equal size.
\newblock In {\em Approximation, randomization, and combinatorial optimization.
  Algorithms and techniques}, pages 218--227. Springer, 2007.

\bibitem{Linial02}
Nathan Linial.
\newblock Finite metric-spaces---combinatorics, geometry and algorithms.
\newblock In {\em Proceedings of the International Congress of Mathematicians,
  Vol. III}, pages 573--586, Beijing, 2002. Higher Ed. Press.

\bibitem{linial1995geometry}
Nathan Linial, Eran London, and Yuri Rabinovich.
\newblock The geometry of graphs and some of its algorithmic applications.
\newblock {\em Combinatorica}, 15(2):215--245, 1995.

\bibitem{LokshtanovMS11-superexp}
Daniel Lokshtanov, D{\'a}niel Marx, and Saket Saurabh.
\newblock Slightly superexponential parameterized problems.
\newblock In {\em Proceedings of the 21st Annual ACM-SIAM Symposium on Discrete
  Algorithms (SODA)}, pages 760--776. SIAM, 2011.

\bibitem{matouvsek2010inapproximability}
Ji{\v{r}}{\'\i} Matou{\v{s}}ek and Anastasios Sidiropoulos.
\newblock Inapproximability for metric embeddings into $ℝ^𝕕$.
\newblock {\em Transactions of the American Mathematical Society},
  362(12):6341--6365, 2010.

\bibitem{nayyeri2015reality}
Amir Nayyeri and Benjamin Raichel.
\newblock Reality distortion: Exact and approximate algorithms for embedding
  into the line.
\newblock In {\em Proceedings of the 56th Annual Symposium on Foundations of
  Computer Science (FOCS)}, pages 729--747. IEEE, 2015.

\bibitem{NayyeriR16}
Amir Nayyeri and Benjamin Raichel.
\newblock A treehouse with custom windows: Minimum distortion embeddings into
  bounded treewidth graphs.
\newblock To appear in SODA 2017.
  http://web.engr.oregonstate.edu/~nayyeria/pubs/tree-dist.pdf, 2016.

\bibitem{papadimitriou2005complexity}
Christos Papadimitriou and Shmuel Safra.
\newblock The complexity of low-distortion embeddings between point sets.
\newblock In {\em Proceedings of the 16th Annual ACM-SIAM Symposium on Discrete
  Algorithms (SODA)}, volume~5, pages 112--118. SIAM, 2005.

\end{thebibliography}


\end{document}